\documentclass[letterpaper,11pt]{article}

\usepackage{booktabs} 
\usepackage[utf8]{inputenc}
\usepackage[margin=1in]{geometry}
\usepackage{bm}
\usepackage{amsfonts,amsmath,amssymb,amsthm}
\usepackage{xcolor}
\usepackage[round]{natbib}

\usepackage{graphicx}
\usepackage{subcaption}

\usepackage{algorithm}

\usepackage{comment}

\newtheorem{lemma}{Lemma}[section]

\newtheorem{Definition}{Definition}[section]

\newtheorem{Example}{Example}
\newtheorem{theorem}{Theorem}[section]

\newcommand{\zhiyi}[1]{\textcolor{blue}{(Zhiyi: #1)}}

\newcommand{\Opt}{\mathrm{Opt}}

\renewcommand{\vec}[1]{\bm{#1}}

\newcommand{\samplenum}{m}

\newcommand{\type}{t}
\newcommand{\typenum}{T}
\newcommand{\typedist}{\mathcal{T}}
\newcommand{\typeset}{[\typenum]}

\newcommand{\val}{v}
\newcommand{\valuenum}{V}
\newcommand{\valuedist}{\mathcal{F}}
\newcommand{\valuedistset}{\mathcal{M}}
\newcommand{\empiricaldist}{\mathcal{E}}

\newcommand{\valueset}{[\valuenum]}

\newcommand{\price}{p}

\newcommand{\segment}{\sigma}
\newcommand{\segmentmap}{\mathcal{G}}
\newcommand{\segmentset}{\Sigma}
\newcommand{\segmentdist}{\mathcal{S}}

\newcommand{\objweight}{\lambda}

\newcommand{\Rev}{\mathrm{Rev}}
\newcommand{\CS}{\mathrm{CS}}
\newcommand{\SW}{\mathrm{SW}}

\newcommand{\twothird}{\tfrac{2}{3}}

\newcommand{\simplex}{\Delta}
\newcommand{\centroid}{\vec{\tau}}
\newcommand{\point}{x}
\newcommand{\pointbelow}{z}

\newcommand{\area}{X}
\newcommand{\areabelow}{Z}
\newcommand{\measure}{\mu}

\newcommand{\R}{\mathbb{R}}

\renewcommand{\Pr}{\mathbf{Pr}}
\newcommand{\E}{\mathbf{E}}

\newcommand{\poly}{\mathrm{poly}}

\newcommand{\twopartdef}[4]
{
	\left\{
	\begin{array}{ll}
		#1 &  #2 \\
		#3 &  #4
	\end{array}
	\right.
}

\begin{document}
\title{\Large Algorithmic Price Discrimination\thanks{To appear in 31st Annual ACM-SIAM Symposium on Discrete Algorithms (SODA), 2020.}}

\author{
Rachel Cummings \thanks{Georgia Institute of Technology. Email: rachelc@gatech.edu. R.C. supported in part by a Mozilla Research Grant, a Google Research Fellowship, and NSF grant CNS-1850187. Part of this work was completed while R.C. was visiting the Simons Institute for the Theory of Computing.}
\and
Nikhil R.\ Devanur \thanks{Microsoft Research. Email: nikdev@microsoft.com}
\and
Zhiyi Huang \thanks{The University of Hong Kong. Email: zhiyi@cs.hku.hk. This work is supported in part by an RGC grant HKU17203717E.}
\and
Xiangning Wang \thanks{The University of Hong Kong. Email: xnwang@cs.hku.hk}
}
\date{}

\begin{titlepage}
\thispagestyle{empty}
\maketitle
\begin{abstract}
\thispagestyle{empty}
We consider a generalization of the third degree price discrimination problem studied in \cite{bergemann2015limits}, 
 where an intermediary between the buyer and the seller 
 can design market segments 
 to maximize any linear combination of consumer surplus and seller revenue. 
Unlike in \cite{bergemann2015limits}, we assume that 
 the intermediary only has partial information about the buyer's value. 
We consider three different models of information, 
 with increasing order of difficulty. 
In the first model, we assume that the intermediary's information 
 allows him to construct 
 a probability distribution of the buyer's value. 
Next we consider the sample complexity model, 
 where we assume that the intermediary only sees samples from this distribution. 
Finally, we consider a bandit online learning model, 
 where the intermediary can only observe past purchasing decisions of the buyer, 
 rather than her exact value. 
For each of these models, we present algorithms 
 to compute optimal or near optimal market segmentation. 

\end{abstract}
\end{titlepage}


\section{Introduction}
\label{sec:intro}

Third degree price discrimination occurs when a seller uses auxiliary information about buyers to offer different prices to different populations, e.g., student and senior discounts 
for movie tickets. 
A modern version of this arises in the context of online platforms that match sellers and buyers. 
Here an intermediary observes information about 
buyers and may pass on some of this information to the seller to help him price discriminate. 
One natural example where price discrimination could be (and often is) used in practice is an \emph{ad exchange}, which matches buyers and sellers of online ad impressions. 
A buyer is an advertiser,
and a seller is a publisher, 
and the impression is sold via an auction 
where the seller sets a reserve price. 
The ad exchange commonly has additional data about the user viewing the impression or about the buyers. 
It could share some of this data with the seller before he sets the reserve price. 


The seminal work of \cite{bergemann2015limits}
shows the following surprising result in such a setting. 
Usually, there is a tradeoff between \emph{social welfare}
which is the value generated by the sale, and seller \emph{revenue}. 
Seller revenue is maximized by setting an appropriate price.  
Social welfare is maximized by selling the item to the buyer as long 
as his value for the item is $\ge 0$, but this generates 0 revenue for the seller. 
Almost magically, \cite{bergemann2015limits} show that 
an intermediary can 
\emph{segment the market} such that it not only maximizes social welfare, 
but also guarantees that the seller revenue doesn't change in the process.
This shows that price discrimination can be used to benefit the customer, 
contrary to the belief that it exploits the customer, 
thus making it palatable.

While this is a strong result,
it requires that the intermediary knows the buyer's exact value, 
which is a very strong assumption, 
and is often not satisfied in practice. 
What is more reasonable is that the intermediary can estimate a personalized probability distribution once the buyer is seen. For instance, if the intermediary observes that the buyer is a student, it may estimate a lower willingness-to-pay, but is unlikely to know the buyer's exact value. Realistically, the intermediary may wish to use machine learning techniques to estimate the personalized probability distribution for a new buyer based upon their observed characteristics and past market data. 
\emph{In this paper, we analyze the power of third degree price discrimination in this setting where the intermediary has only a noisy signal of a buyer's value. 
}

\vspace{-3pt}
\subsection{Model and Results}
The seller sells a single item, and 
there is a single buyer.
We consider value distributions with a finite support.
We assume that the intermediary observes finitely many types of buyers;
each type is associated with a different distribution over the values. 
We denote the set of values by $\valueset= \{ 1, 2, \dots, \valuenum \}$,  
the set of types by $\typeset= \{ 1, 2, \dots, \typenum \}$, and 
the distribution over values given a type $\type$ by $\valuedist(\type)$.
We denote the distribution over types by $\typedist$. 
Given this, the mechanism proceeds as follows.  This is illustrated in Figure \ref{fig.timeline}.
\begin{enumerate}
	\item  A segmentation is a pair of
	a segment set $\segmentset,$ and  
	a segment map $\segmentmap : \typeset \mapsto \simplex(\segmentset ),$
	where $\simplex(\cdot)$ denotes the set of all probability distributions over a given domain.   
	Once the intermediary decides on a segmentation, 
	it is revealed to the seller.
	\item When a buyer arrives, her type $\type$ and value $\val$ are drawn from the prior distribution. The intermediary observes only her type $\type$
	but not the value $\val$.
	\item Intermediary draws a segment $\segment$ from the distribution $\segmentmap(\type)$ 
	and reveals it to the seller.   
	\item On observing a segment $\segment$, the seller posts 
	the monopoly price $\price$ for the value distribution conditioned on observing $\segment$. 
	\item Buyer buys the item if and only if her value $ \val \ge \price$. 
\end{enumerate}

The model in \cite{bergemann2015limits} is the special case where 
the type set is identical to the value set, and 
the distributions $\valuedist(\type)$ are point masses.

\begin{figure*}
	\centering
	\includegraphics[width=.9\textwidth]{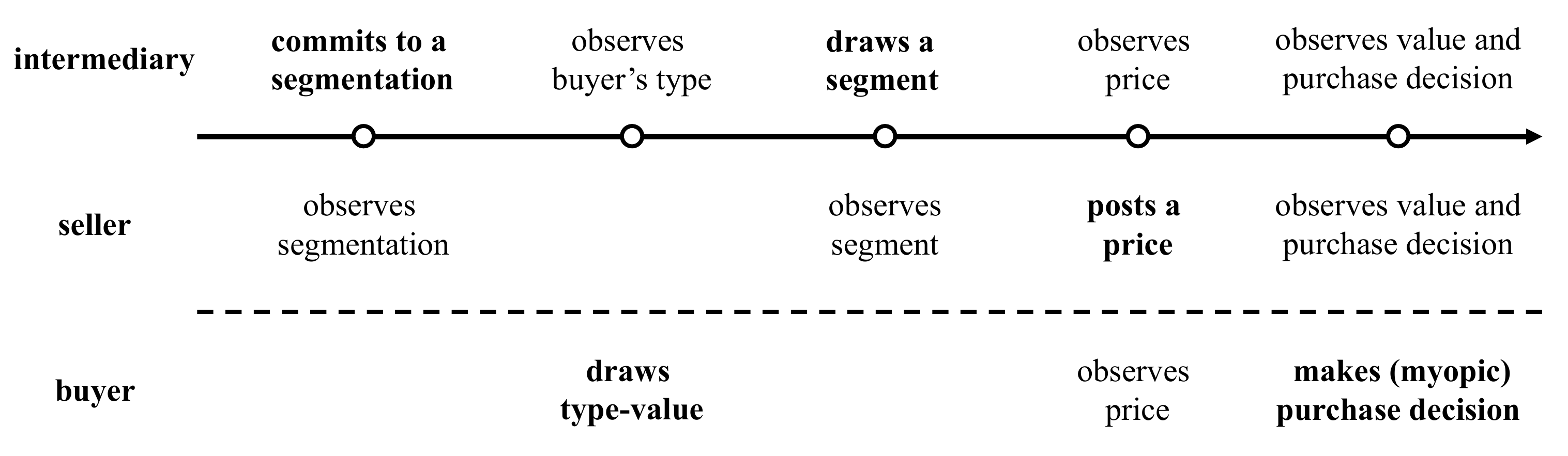}
	\caption{Time line of a single round.}\label{fig.timeline}
\end{figure*}

A price is said to be a \emph{monopoly price} if it maximizes seller revenue 
for a given distribution, 
and this revenue is called the \emph{monopoly revenue}. 
We call the marginal distribution over values $\valueset$ as the \emph{prior} distribution. 
The seller is always guaranteed at least the monopoly revenue for the prior distribution, 
since he can ignore the segment information and set a monopoly price. 

The intermediary's objective is some given positive linear combination of 
seller revenue and \emph{consumer surplus}. 
Consumer surplus is the expectation of the buyer's utility,
which is $v-p$ if the buyer  with value $v$ buys the item at price $p$, 
and is 0 otherwise.
Of particular interest is the special case of
maximizing consumer surplus alone. 
We consider three informational models of increasing difficulty for the intermediary, and show the following results.   

\paragraph{Bayesian:} The intermediary and the seller know the value-type distributions: $\valuedist(\type)$ for all $\type$, and $\typedist$. 
We show that the optimal segmentation can be computed using a linear program (LP). 
The range of achievable values for consumer surplus and revenue depend on the distribution, and one may not always be able to achieve the full consumer surplus as in \cite{bergemann2015limits}. 
Some other nice properties may not hold as well, see Appendix \ref{sec:example} for examples. 

\paragraph{Sample Complexity:} The intermediary and the seller  observe a batch of signal-value pairs sampled from the underlying distribution. We are interested in the number of samples required  to get an $\epsilon$ approximation. 
We first cosntruct a distribution for which no bounded function of $\epsilon$ is sufficient. The $\valuedist(\type)$'s in this example satisfy both boundedness 
and  regularity, which are standard assumptions in the sample complexity of mechanism design. 
This  points to further limitations on what such an intermediary can do: 
in case of noisy signals, the distribution cannot be arbitrary. 
Motivated by this, 
we identify a property about the distributions, 
which we call MHR-like,  
and show (via an algorithmic construction) that a polynomial number of samples are sufficient. This is the technically most challenging part of the paper and most of the focus in the main body of the paper is on this part. 

\paragraph{\bf Online Learning:} The intermediary must learn the segmentation online using only bandit feedback from the buyer's decision to purchase or not at the seller's chosen price.  
The last step of the timeline depicted in Figure \ref{fig.timeline} is modified in this setting so that the intermediary and seller only observe the buyer's purchase decision, not her value.
Here we give \emph{no-regret} learning algorithms.
Clearly, we need certain assumptions on the seller's behavior for any nontrivial result; 
there is not much we can do if the seller picks prices randomly all the time.
Our assumptions can accommodate natural no regret learning algorithms on the seller side, including the Upper-Confidence-Bound (UCB) algorithm and the Explore-then-Commit (ETC) algorithm.

\subsection{Contributions to the Sample Complexity of Mechanism Design}

Pioneered by \cite{BBHM05}, \cite{Elk07}, and \cite{DRY15}, and formalized by \cite{CR14}, the sample complexity of mechanism design, in particular, the revenue maximization problem, has been a focal point in algorithmic game theory in the last few years \cite{MR14, BSV16, DHP16, MR16, HT16, CD17, GN17, GW18, HuangMR/2018/SICOMP, GHZ19}.

This paper adds to the literature of sample complexity of mechanism design in two-folds.
The first one is conceptual: 
we formulate the first sample complexity problem from the viewpoint of an \emph{intermediary} rather than the seller, and for the task of designing \emph{information dispersion} rather than allocations and payments.
We show impossibility results for the general case and, more importantly, identify sufficient conditions under which we derive positive algorithmic results.

Conceptually new models often lead to new technical challenges.
Our second contribution is an algorithmic ingredient that tackles such a new challenge.
Let us start with a thought experiment:
consider a more powerful intermediary who knows the true distributions;
the seller, however, still acts according to some beliefs formed from the observed samples.
\emph{Does the problem become trivial?
Can the intermediary simply run the optimal segmentation w.r.t.\ the true distributions and expect near optimal outcomes?}

The answers turn out to be negative.
Consider a segment for which there are two prices $p^*$ and $p$, such that $p^*$ is the monopoly price with a sale probability close to $1$, while $p$ gets near optimal revenue with a sale probability close to $0$. 
If the intermediary includes this segment, however, the seller's beliefs may overestimate the revenue of $p$ and/or underestimate that of $p^*$ and, thus, deduce that $p$ is the monopoly price instead of $p^*$.
As a result, the resulting social welfare may be much smaller than what the intermediary expects from the true distributions.

This example shows that, unlike existing works on the sample complexity of mechanism design, where the difficulties arise purely from the learning perspective of the problem, our problem presents an extra challenge from the uncertainties in the seller's behavior due to his inaccurate beliefs. 

Intuitively, the intermediary would like to convert the optimal segmentation w.r.t.\ the true distributions into a more robust version, such that for any approximately accurate beliefs that the seller may have, the resulting objective is always close to optimal.
We will refer to this procedure as \emph{robustification}, and the result as the \emph{robustified segmentation}.

\section{Bayesian Model}
\label{sec:bayesian}
We start with an example through which we will illustrate the main ideas in this section. 

\begin{Example}
The example is parameterized by a noise level, $1-z \in [0,1]$. 
 The value set $\valueset = \{1,2,3\}$ is identical to the type set $\typeset$. 
Each type $\type \in \typeset$ corresponds to the distribution 
$\valuedist(\type)$: 
$$
\Pr_{\val \sim \valuedist(t)} [\val] = \twopartdef{ z }{\text{ if } \val = \type,  } { \tfrac{1-z}{2}}{\text{otherwise.}}
$$
When $z=1/3$, all $\valuedist(t)$'s equal the uniform prior, and no non-trivial segmentation is possible. At the other extreme, when $z=1$ each $\valuedist(t)$ is a pointmass at $t$, which is the \cite{bergemann2015limits} model. 
\end{Example}

\paragraph{Simplex View.} 
As observed by \cite{bergemann2015limits}, 
the key idea is to identify segments with probability distributions over $\valueset$. 
The only thing that matters given a segment $\segment$ is 
the posterior distribution on $\valueset$ 
conditioned on the intermediary choosing $\segment$. 
Since $\valueset$ is finite, 
it is easier to think of $\simplex(\valueset)$ as the unit simplex 
in the appropriate dimensions. 
Then, a segment $\segment$ is simply a point in this simplex. 
Further, all that matters for a segmentation is 
the distribution over $\segment$'s as observed by the seller, i.e., 
it is sufficient to specify a distribution over the simplex $\simplex(\valueset)$. 
The only constraint on this distribution is that its expectation must
equal the prior distribution over values, 
which is another point on the simplex, denoted by $\centroid$.

Going further, it is sufficient to only consider 
some special points on the simplex. 
We denote these special 
points by $x^S$, for a subset $S \subseteq \valueset$: 
this is the equal revenue distribution with support equal to the set $S$, i.e., 
these are distributions supported on $S$ such that 
$ p\cdot \Pr[v\ge p] $ is the same for all $p\in S$. 
These special points partition the simplex into 
regions $\area_\val$ for each $\val \in \valueset$: 
each distribution in $\area_\val$ is such that 
$\val$ is a monopoly price for it. 

\begin{figure*}
	\centering
	\begin{subfigure} {.43\textwidth}
		\includegraphics[width=\textwidth]{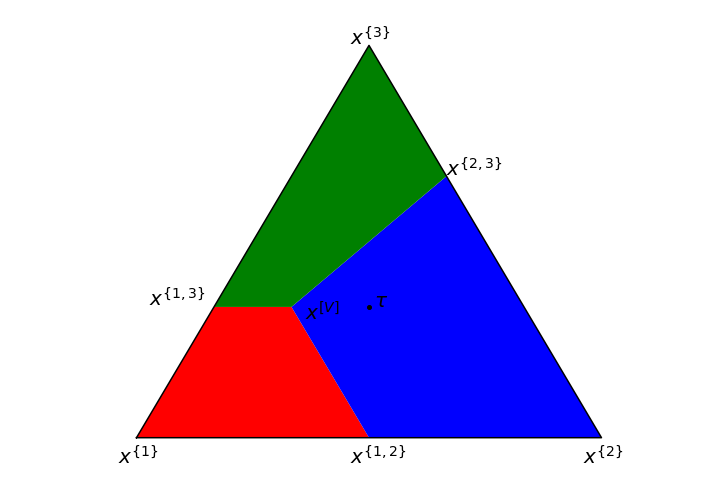}
		\caption{Partition of the simplex into $\area_\val$'s and\\ $\tau =$ the uniform distribution over $\valueset.$\\~ }
		\label{fig:simplex1} 
	\end{subfigure}
	\begin{subfigure}{.43\textwidth}
		\centering
		\includegraphics[width=\textwidth]{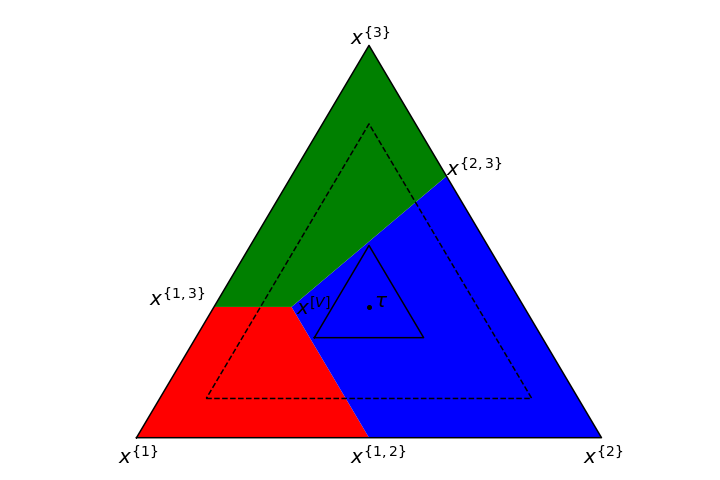}
		\caption{Convex hull of types for $z<1$. Solid inner triangle is for $z=0.49$; dashed one is for $z=0.8$.}
		\label{fig:simplex2} 
	\end{subfigure}	
	\caption{Simplex view for Example 1}
	\label{fig:simplex}
\end{figure*}

We now describe this through Example 1. 
Figure \ref{fig:simplex1} shows 
the unit simplex, and the points $x^S$ for all $S \subseteq \valueset$. 
The red region is $X_1$, blue is $X_2$ and green is $X_3$. 
The uniform distribution over $\valueset$ is represented by $\centroid$. 
An optimal segmentation with no noise (when $z=1$) 
corresponds to representing $\centroid$ as 
the following convex combination of the vertices of the blue polytope: 
\[ \centroid =   \tfrac{1}{6} x^{\{2\}} +  \tfrac{1}{6} x^{\{2,3\}} + \twothird x^{\valueset} .\]
This corresponds to the following segmentation with $\segmentset = \{\sigma_1,\sigma_2,\sigma_3\}$ corresponding to $x^{\{2\}} , x^{\{2,3\}}$ and $ x^{\valueset}$ resp. and $\segmentmap(1) = (0,0,1)$; $\segmentmap(2) = (1/2,1/6,1/3)$; and $\segmentmap(3) = (0,1/3,2/3)$. 

\cite{bergemann2015limits} consider segmentations that only consist of vertices of $X_{p^*}$ 
where $p^*$ is a monopoly price of the prior distribution 
(which is the blue region $X_2$ in Figure \ref{fig:simplex1}). 
They assume that ties for monopoly price are broken in favor of the lowest price, which is the lowest value in the support.
This implies that the item is always sold thus maximizing social welfare.

\paragraph{Generalization of the Simplex View.}
We extend this simplex view to our model. 
When the number of types $\typenum$ is at most the number of values $\valuenum$, 
and the type distributions are non-degenerate, 
we can continue to consider the simplex on the set of values, 
as we have done so far. 
This will be the case for Example 1.
A more general view is to consider the simplex on the set of types.  
Most of the intuition extends to this view, 
although geometrically the picture is somewhat different. 
(This is the view we use in the proofs; 
the simplex on values is used just for illustration.) 

The case when $z<1$ is depicted in Figure \ref{fig:simplex2}. 
The main difference from the previous picture is that 
we are not allowed to choose any point on the simplex for our segmentation. 
Instead, we are restricted to only choose the points in the convex hull of the $\valuedist(\type)$s, for all $\type \in \typeset$.
We denote this convex hull by $\simplex(\typeset)$, by abuse of notation.  
For $z = 0.49$, the figure shows that $\simplex(\typeset)$ is 
contained entirely inside the blue region. 
Thus no matter what segmentation is used, 
the seller always sets the monopoly price of 2; 
segmentation is therefore useless.  
For $z=0.8$, the figure shows that segmentation is possible because $\simplex(\typeset)$ intersects with all three regions, $X_1,X_2,X_3$. 


We introduce some notation now. 
Given any segmentation $(\segmentset, \segmentmap)$, 
this induces a distribution over segments, denoted $\segmentdist$, 
and a posterior distribution on the values $\valueset$ 
for each segment $\segment \in \segmentset$, 
which we abuse notation and denote by $\valuedist(\segment)$.
For any distribution $\valuedist$, let $\Rev(\valuedist)$ denote its monopoly revenue.
Let $\CS(\valuedist)$ denote the consumer surplus when the seller sets the monopoly price for distribution $\valuedist$.	
Our goal is to find a segmentation to maximize a linear combination of revenue and consumer surplus, i.e., for some parameter $\objweight\in [0,1] $, maximize:
%
\[
\E_{\segment \sim \segmentdist} \bigg[ \objweight \cdot \Rev\big(\valuedist(\segment)\big) +  (1-\objweight) \cdot \CS\big(\valuedist(\segment)\big) \bigg]
~.
\]

From now on, we let $\simplex(\typeset) = \{ \vec{\point} \in \R^T_+ : \| \vec{\point} \|_1 = 1 \}$ denote probability distributions over the types. 
We first formalize the claim that segmentation schemes correspond to 
probability distributions over $\simplex(\typeset)$ with a given expectation. 
The proofs in the section are deferred to Appendix~\ref{app:bayesian-proof}.

\begin{lemma}
	\label{thm:simplex-view}
	Let $\centroid$ denote the point in the simplex $\simplex(\typeset)$ corresponding to the distribution $\typedist$: $\centroid = (\Pr_{\typedist}[1], \Pr_{\typedist}[2], \dots, \Pr_{\typedist}[T])$. 
	There is a 1:1 correspondence between segmentations   $(\segmentset, \segmentmap)$ and probability distributions $\measure$ over $\simplex(\typeset)$ such that the expectation is $\centroid$, i.e., 
	\begin{equation}
	\label{eqn:centroid_constraint}
	\int \vec{\point} d\measure  = \centroid .
	\end{equation}
\end{lemma}
Using this lemma, 
we switch our design space to probability distributions $\measure$ that satisfy \eqref{eqn:centroid_constraint}. 
We use $\segment \in \segmentset$ and $\vec{\point} \in \simplex(\typeset)$ interchangably. 


We now partition the simplex $\simplex(\typeset)$ into $\valuenum$ areas, $\area_1, \area_2, \dots, \area_\valuenum$, one for each value/price in $\valueset$, such that the price $p$ is a monopoly price for any segment in $\area_p$.
For any distribution $\valuedist$ and any price $p \in \valueset$, let $\Rev(\valuedist, p)$ denote the revenue of price $p$ on distribution $\valuedist$, 
and $\CS(\valuedist, p)$ denote the consumer surplus. 
For any $p \in \valueset$, define:
\[
\area_p = \{ \vec{\point} : \Rev(\valuedist(\vec{\point}), p) \ge \Rev(\valuedist(\vec{\point}), p'), \forall p' \ne p \}
~.
\]
%
Since the revenue function $\Rev(\valuedist(\vec{\point}), p) = \sum_{\type} \point_\type \cdot \Rev(\valuedist(\type), p))$
is linear in $\vec{\point}$, 
the set $\area_p$ is the intersection of the simplex $\simplex(\typeset)$ 
and a polytope defined by $\valuenum - 1$ linear constraints.
Further, if we restrict our domain to points $\vec{\point} \in \area_p$, we have that
these are linear functions in $\vec{\point}.$
\[
\Rev(\valuedist(\vec{\point})) = \Rev(\valuedist(\vec{\point}), p) = \sum_{\type} \point_\type \cdot \Rev(\valuedist(\type), p)
~,\text{ and } 
\]
\[
\CS(\valuedist(\vec{\point})) = \CS(\valuedist(\vec{\point}), p) = \sum_{\type} \point_\type \cdot \CS(\valuedist(\type), p) ~.
~
\]
We next observe that this implies that it is sufficient to choose at most one point from each $\area_p$. 
The idea is that we can replace the distribution conditioned on $\vec{x} \in \area_p$ by its expectation. 

\begin{lemma}
    \label{lem:bayesian-finite-support}
	There is an optimal segmentation such that the distribution $\measure$ is supported on at most one point from each $\area_p$, i.e., a finite set of the form $\{ \vec{\point}_p \in \area_p,  \forall ~p \in \valueset\}$. 
\end{lemma}

Using this lemma, we now show that the following linear program (LP) captures the optimal segmentation. 
The variables are $\vec{\pointbelow}_p = \vec{\point}_p \cdot \measure(\vec{\point}_p )$. 
We denote by $\areabelow_p$ the region that is the convex hull of $\area_p$ and the origin.  
$\Rev$ and $\CS$ extend naturally to $\areabelow_p$s. 
\begin{align}
\max \quad & \sum_{p \in \valueset}  \objweight \cdot \Rev\big(\valuedist(\vec{\pointbelow}_p)\big) +  (1-\objweight) \cdot \CS\big(\valuedist(\vec{\pointbelow}_p)\big) \label{eq:LP} \\
\text{s.t.} \quad &\forall ~ p \in \valueset, \vec{\pointbelow}_p \in \areabelow_p 
\text{ and }
\sum_{p \in \valueset} \vec{\pointbelow}_p = \centroid 
~. \nonumber
\end{align}
\begin{theorem}
    \label{thm:bayesian}
    We can find an optimal segmentation in polynomial time by solving LP~\eqref{eq:LP}.
\end{theorem}

\section{Sample Complexity Model}
\label{sec:sample}

We scale the values to be in $(0, 1]$, i.e., $\valueset = \big\{ \frac{1}{\valuenum}, \frac{2}{\valuenum}, \dots, 1 \big\}$.
This treatment simplifies the notations in the proofs, and separates the two roles of $\valuenum$: the scale of the values (less interesting, always has the same degree as $\epsilon$), and the number of possible values.
To translate the bounds into the original scaling, replace $\epsilon$ with $\frac{\epsilon}{\valuenum}$ everywhere.
We further assume the type distribution to be uniform to simplify discussions. 
This is w.l.o.g.\ up to duplication of types.

Following standard notations in algorithmic mechanism design, we refer to the sale probability of a price as its quantile.
We will consider the revenue curve in the quantile space where the $x$ and $y$ coordinates are the quantile of a price and its revenue, respectively.


\subsection{Model and Results}

\paragraph{Intermediary:}
The intermediary has access to the value distributions of different types only in the form of $\samplenum$ i.i.d.\ samples per type.
She chooses a segmentation based on these samples, and then the chosen segmentation is evaluated on a freshly drawn type-value pair, i.e., the test sample.
The expectation of the objective is taken over the random realization of the $m$ samples per type as well as the test sample, and potentially the randomness in the choice of the segmentation.

\paragraph{Buyer:}
%
The buyer bids truthfully since the seller effectively posts a take-it-or-leave-it price.

\paragraph{Seller:}
%
We need to further define how the seller acts.
Consider the following candidate models:
\begin{enumerate}
    \item The seller knows the value distributions exactly. Hence, given the segmentation and the realized segment, which induces a mixture of the value distributions of different types, the seller posts the monopoly price of the mixture.
    \item The seller can access the same set of $\samplenum$ samples per type, and believes that the value distributions are the empirical distributions, i.e., the uniform distributions over the corresponding samples.
    Hence, she posts the monopoly price of the mixture of empirical distributions.
    \item The seller further has access to other sources of samples.
    \item The seller further has access to other sources of prior knowledge.
    %
\end{enumerate}

This is only a nonexclusive list of many potential models that are equally well-motivated in our opinion, depending on the actual applications.
\emph{Is there a unifying model that allows us to study all these settings in one shot and get non-trivial positive results?}

To this end, this paper considers the following overarching model (the subscript $S$ indicates that these variables are associated with the \emph{seller}):
\begin{quote}
    For $\epsilon_S = O \big( \samplenum^{-1/2} \log(\samplenum \valuenum) \big)$, the seller forms beliefs $\valuedist_S(\type)$'s, $\type \in \typeset$, such that for any type $\type$ the Kolmogorov-Smirnov distance between $\valuedist_S(\type)$ and $\valuedist(\type)$ is at most $\epsilon_S$, i.e., for any value $v \in \valueset$, $v$'s quantiles w.r.t.\ $\valuedist(\type)$ and $\valuedist_S(\type)$ differ by at most $\epsilon_S$.
    Then, she posts the monopoly price of the mixture of the beliefs.
\end{quote}

The choice of $\epsilon_S$ is based on a standard concentration plus union bound combination on the empirical distributions over the $\samplenum$ samples that the intermediary can access.
In other words, we assume that the seller's beliefs are at least as good as what could have been estimated using the intermediary's samples.
All aforementioned candidate models are special cases of ours.

We start with an impossibility result for general value distributions. 
See Section~\ref{sec:example} for details.

\begin{theorem}
    \label{thm:sample-impossibility}
    If the value distributions are allowed to have multiple monopoly prices whose social welfare differ by at least $\Omega(1)$, e.g., the uniform distribution over $\{\tfrac{1}{2}, 1\}$, no algorithm can obtain any $o(1)$-approximation using a bounded number of samples.
\end{theorem}

\paragraph{(Discrete) MHR-like Distributions.}
Given the above impossibility result that relies on value distributions that have multiple monopoly prices whose respective values of social welfare are vastly different,
we intuitively need the value distributions to be unimodal and far from having a plateau.
The family of continuous monotone hazard rate (MHR) distributions, a standard family of distributions in the literature, has all the nice properties that we need, except that they are continuous.
They are unimodal since they have concave revenue curves in the quantile space (folklore).
In fact, their revenue curves in the quantile space are strongly concave near the monopoly price \cite[Lemma~3.3]{HuangMR/2018/SICOMP}.
They also admit other useful properties: the optimal revenue is at least a constant fraction of the social welfare \cite[Lemma~3.10]{DRY15}; and the monopoly price has a sale probability lower bounded by some constant \cite[Lemma~4.1]{HMS08}.

There is an existing notion of discrete MHR distributions by~\cite{BMP63} that mimics the functional form of the continuous version.
However, it loses some useful properties.
In particular, it contains some distributions that have two monopoly prices, e.g., the uniform distribution over $\{\tfrac{1}{2}, 1\}$, and as a result still suffers from the impossibility result.

Instead, we define a family of (discrete) MHR-like distributions directly from the aforementioned benign properties of continuous MHR distributions.
Hence, unlike the existing notion of discrete MHR distributions, our definition truly inherits the main features of continuous MHR distributions.
We remark that the constants $\frac{1}{4}$ and $\frac{1}{e}$ in the following definition are merely copied from the continuous counterparts; our results still hold asymptotically if they are replaced by other constants.

\begin{Definition}[MHR-like Distributions]
    \label{def:mhr-like-distributions}
    A discrete distribution $\valuedist$ is \emph{MHR-like} if it satisfies:
    \begin{enumerate}
        \item \label{con:concavity}
        (Concavity) Its revenue curve is concave in the quantile space.
        \item \label{con:strong-concavity}
        (Strong concavity near monopoly price) For its monopoly price $p^*$ and any other price $p'$, suppose their quantiles are $q^*$ and $q'$ respectively; then, we have:
        \[
            \Rev \big( p', \valuedist \big) \le \big( 1 - \tfrac{1}{4}(q^* - q')^2 \big) \cdot \Rev \big( p^*, \valuedist \big)
            ~.
        \]
        \item \label{con:monopoly-sale-prob}
        (Large monopoly sale probability) Its monopoly price's sale probability is at least $\frac{1}{e}$.
        \item \label{con:revenue-welfare-ratio}
        (Small revenue and welfare gap) Its monopoly revenue is at least $\frac{1}{e} \E_{v \sim \valuedist} [v]$.
    \end{enumerate}
\end{Definition}

The main difference of our MHR-like distribution and the notion in~\cite{BMP63} is property 2 in definition~\ref{def:mhr-like-distributions}. An MHR-like distribution can be made by discretizing an continuous MHR distribution, meanwhile ensuring that there is a gap between the optimal revenue and any sub-optimal ones.  

We show that polynomially many samples are sufficient for learning an $\epsilon$-optimal segmentation, with only the mild assumption on seller's behavior discussed earlier in the section.

\begin{theorem}
    \label{thm:sample-complexity-mhr-like}
    With $\samplenum = \poly \big( \epsilon^{-1}, \typenum, \log \valuenum \big)$ i.i.d.\ samples, we can learn a segmentation that is optimal up to an $\epsilon$ additive factor in $\poly \big( \epsilon^{-1}, \typenum, \valuenum \big)$ time.
\end{theorem}

\subsection{Robustification: Motivation and Definition}
\begin{figure*}
    \centering
    \begin{subfigure}{.43\textwidth}
        \includegraphics[width=\textwidth]{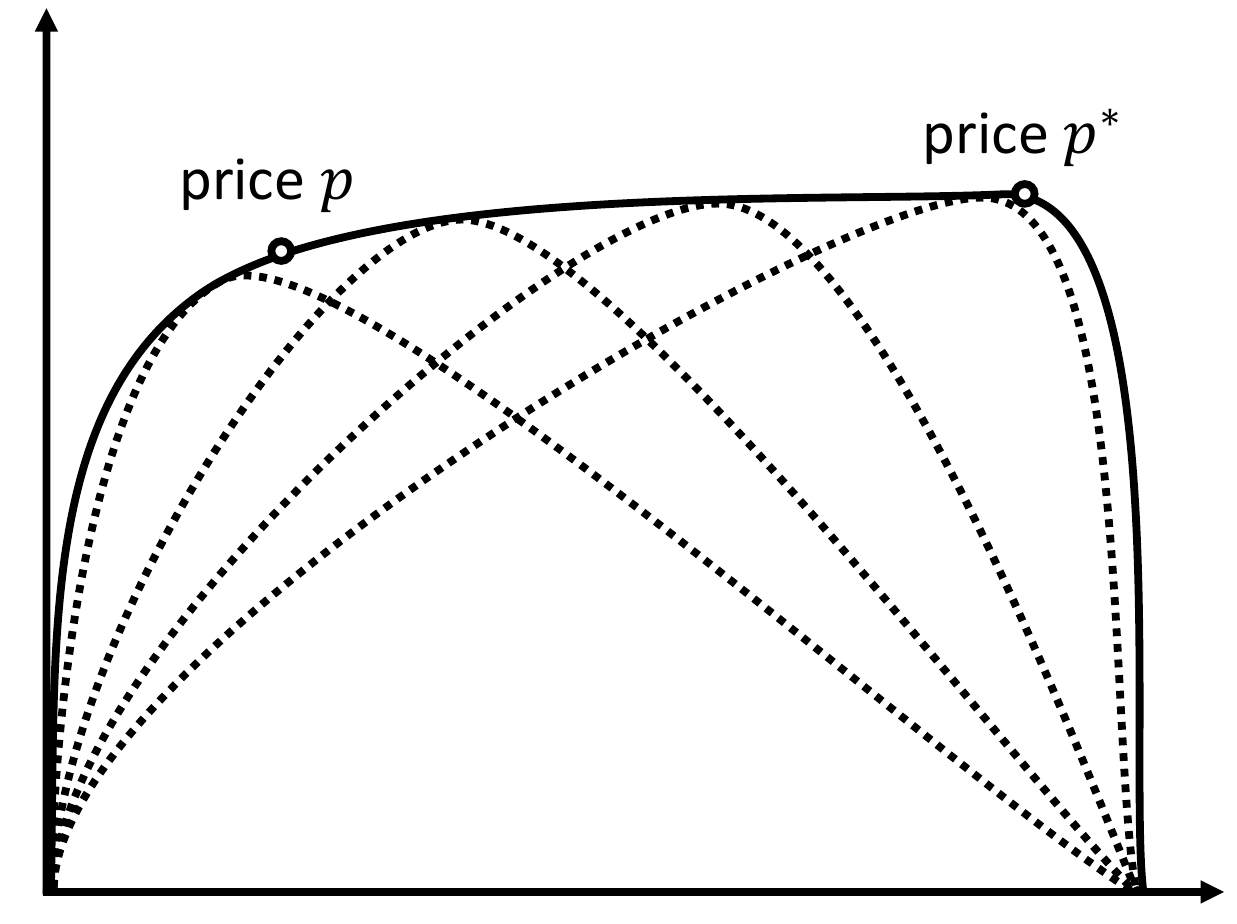}
        \caption{True distribution}
        \label{fig:plateau-true}
    \end{subfigure}
    \begin{subfigure}{.43\textwidth}
        \includegraphics[width=\textwidth]{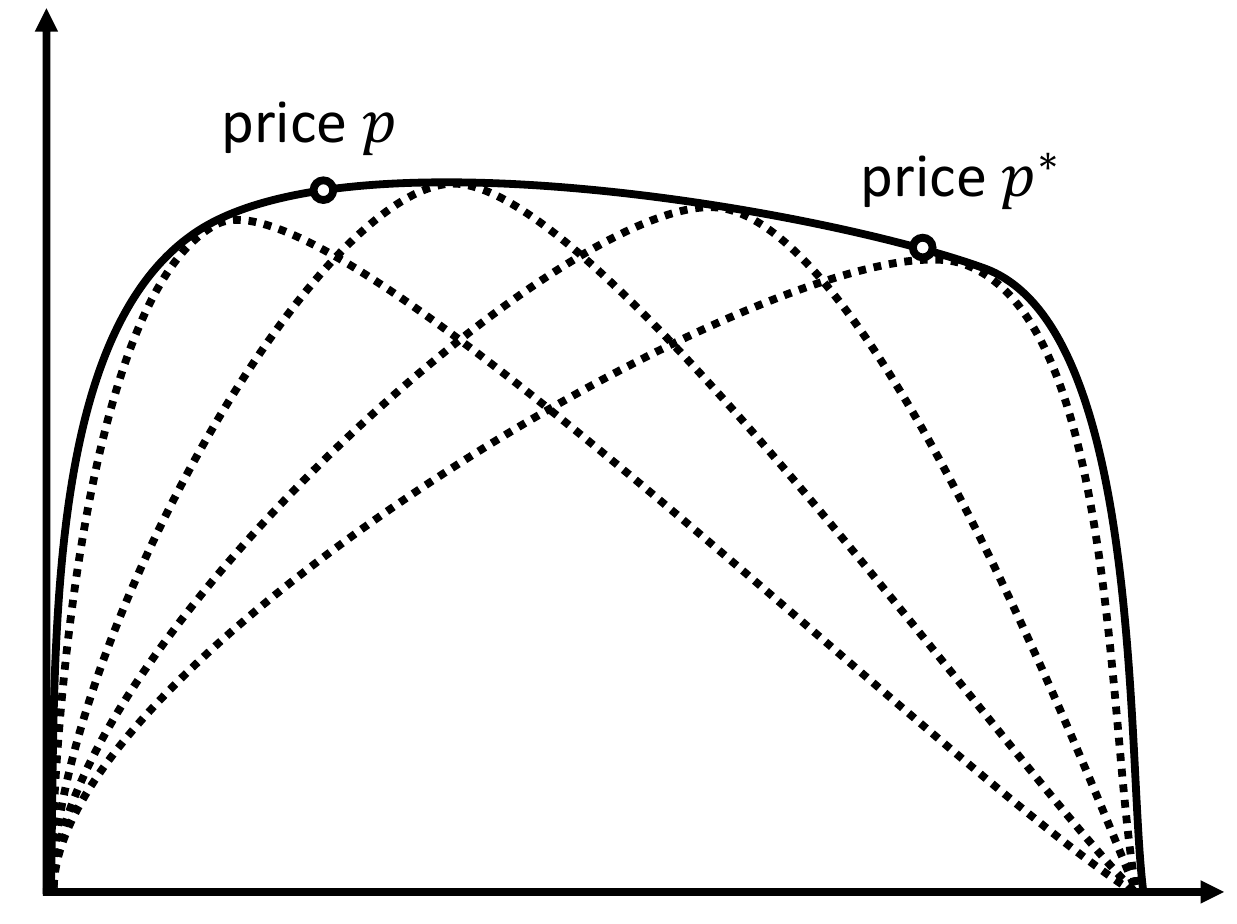}
        \caption{Seller's belief}
        \label{fig:plateau-seller}
    \end{subfigure}
    \caption{Plateau Example. In a revenue curve in the quantile space, the $x$ and $y$ coordinates are the quantile of a price and its revenue respectively. On the left, the solid curve is the revenue curve of a segment w.r.t.\ the true distributions; the dotted curves are those of the type distributions mixed in the segment. On the right are the counterparts w.r.t.\ the seller's beliefs. Prices $p^*$ and $p$ are monopoly prices of the segment w.r.t.\ the true distributions and the seller's beliefs respectively.}
    \label{fig:plateau}
\end{figure*}
Recall the thought experiment in Section~\ref{sec:intro}.
Consider a more powerful intermediary who has exact knowledge of the true distributions;
the seller, however, still acts according to her approximately accurate beliefs.
Further, recall the example where the problem remains nontrivial even when the intermediary has more power; 
we give more details below.
Consider a segment for which there are two prices $p^*$ and $p$, such that $p^*$ is the monopoly price with a quantile close to $1$, while $p$ gets close-to-optimal revenue with a quantile close to $0$. 
Even though a single MHR-like distribution cannot have such a plateau, a mixture of MHR-like distributions can.\footnote{In fact, every distribution on $[\valuenum]$ is a mixture of MHR-like distributions, because point masses are MHR-like.}
See Figure~\ref{fig:plateau-true} for an illustrative example.
If the intermediary includes this segment, 
the seller's belief may overestimate the revenue of $p$ 
and/or underestimate that of $p^*$ and thus 
deduce that $p$ is the monopoly price instead of $p^*$.
See Figure~\ref{fig:plateau-seller}.
The resulting social welfare may be much smaller than what the intermediary expects from the true distributions.


Intuitively, we would like to convert the optimal segmentation w.r.t.\ the true distributions into a more robust version, such that for any approximately accurate beliefs that the seller may have, the resulting social welfare and revenue are both close to optimal.
As mentioned in Section~\ref{sec:intro}, we will refer to this procedure as \emph{robustification}, and the result as the \emph{robustified segmentation}.

In the following definition, the subscripts of $\epsilon_S$ and $\epsilon_I$ indicate they are the additive errors that the seller and intermediary are aiming for, respectively.
Further, it states that the robustified segmentation must keep all segments in the original version ($\segmentset^r \supseteq \segmentset^*$).
We ignore insignificant segments due to technical difficulties in achieving the stated properties for them,%
\footnote{If the expected value is tiny in the first place, all prices are $\epsilon_S$-optimal. Hence, we cannot achieve robustness.}
and that their roles in the revenue and social welfare are negligible.
For any significant segment, the first two conditions state that its weight and mixture of types are preserved approximately; 
the third condition gives the desirable robustness against the uncertainties in the seller's behavior.

\begin{Definition}[Robustified Segmentation]
    \label{def:robust-segmentation}
    Suppose that (1) $(\segmentset^*, \segmentmap^*)$ is a segmentation, represented by $\vec{\point}^*_\segment$ and weight $w^*_\segment=  \measure(\vec{\point}^*_\segment)$, $\forall~\segment \in \segmentset^*$; and (2) $p^*_\segment$ is an optimal price w.r.t.\ $\valuedist(\vec{\point}^*_\segment)$,  $\forall~\segment \in \segmentset^*$. 
    For any $\epsilon_I \ge \epsilon_S > 0$, $(\segmentset^r, \segmentmap^r)$ is an $(\epsilon_I, \epsilon_S)$-robustified segmentation, represented by $\vec{\point}^r_\segment$ and weight $w^r_\segment$ with $\segmentset^r \supseteq \segmentset^*$, if for any $\segment \in \segmentset^*$, either $\segment$ is insignificant in that $\E_{v \sim \valuedist(\vec{\point}^*_\segment)} [v] < \epsilon_I$, or: 
    \begin{enumerate}
        \item (Weight preservation) $w^r_\segment \ge (1 - \epsilon_I) \cdot w^*_\segment$;
        \item (Mixture preservation) $\| \vec{\point}^*_\sigma - \vec{\point}^r_\sigma \|_1 \le \epsilon_I$; and
        \item (Robustness) no $\epsilon_S$-optimal price w.r.t.\ $\valuedist(\vec{\point}^r_\sigma)$ has a quantile smaller than that of $p^*_\segment$ by $\epsilon_I$.
    \end{enumerate}
\end{Definition}

The next lemma shows that the technical conditions in the definition of robustified segmentation indeed lead to robust bounds in terms of both social welfare and revenue and, by induction, their linear combinations.
The proof follows from straightforward calculations and therefore is deferred Appendix \ref{app:sampleproof}.

\begin{lemma}
    \label{lem:robust-segmentation-performance-main}
    For any prices $p^r_\segment$'s that are $\epsilon_S$-optimal w.r.t.\ segments $\segment \in \segmentset^r$ in the robustified segmentation, we have the following in terms of social welfare and revenue: 
    \begin{align*}
        &    
        \textstyle
        \sum_{\segment \in \segmentset^r} w^r_\segment \cdot \SW \big( p^r_\segment, \valuedist(\vec{\point}^r_\segment) \big) 
        \ge 
          \sum_{\segment \in \segmentset^*} w^*_\segment \cdot \SW \big( p^*_\segment, \valuedist(\vec{\point}^*_\segment) \big) - O\big(\epsilon_I\big)
        ~, \\
        &    
        \textstyle
        | \sum_{\segment \in \segmentset^r} w^r_\segment \cdot \Rev \big( p^r_\segment, \valuedist(\vec{\point}^r_\segment) \big) 
        - 
         \sum_{\segment \in \segmentset^*} w^*_\segment \cdot \Rev \big( p^*_\segment, \valuedist(\vec{\point}^*_\segment) \big) | \le O\big(\epsilon_I\big) 
        ~.
    \end{align*}
\end{lemma}

\subsection{Robustification: Algorithm}
\label{sec:robustification-algorithm}

This subsection introduces an algorithm that finds such an $(\epsilon_I, \epsilon_S)$-robustified segmentation in polynomial time for any sufficiently large $\epsilon_I$, i.e.:
\begin{equation}
    \label{eqn:robustify-epsilons}
    \epsilon_I \ge \epsilon_S^{\frac{1}{6}} \typenum^{\frac{2}{3}} \log^{\frac{1}{6}} \valuenum = \tilde{O} \left( \samplenum^{-\frac{1}{12}} \typenum^{\frac{2}{3}} \right)
    ~,
\end{equation}

\begin{lemma}
    \label{lem:robustify}
    There is an algorithm that computes in polynomial time an $(\epsilon_I, \epsilon_S)$-robustified segmentation, for any $\epsilon_I$ and $\epsilon_S$ that satisfy Eqn.~\eqref{eqn:robustify-epsilons}.
\end{lemma}


\subsubsection{Proof Sketch of Lemma~\ref{lem:robustify}}
\label{sec:proof-robustify-main}

%
\begin{figure*}
    \centering
    \begin{subfigure}{.43\textwidth}
        \includegraphics[width=\textwidth]{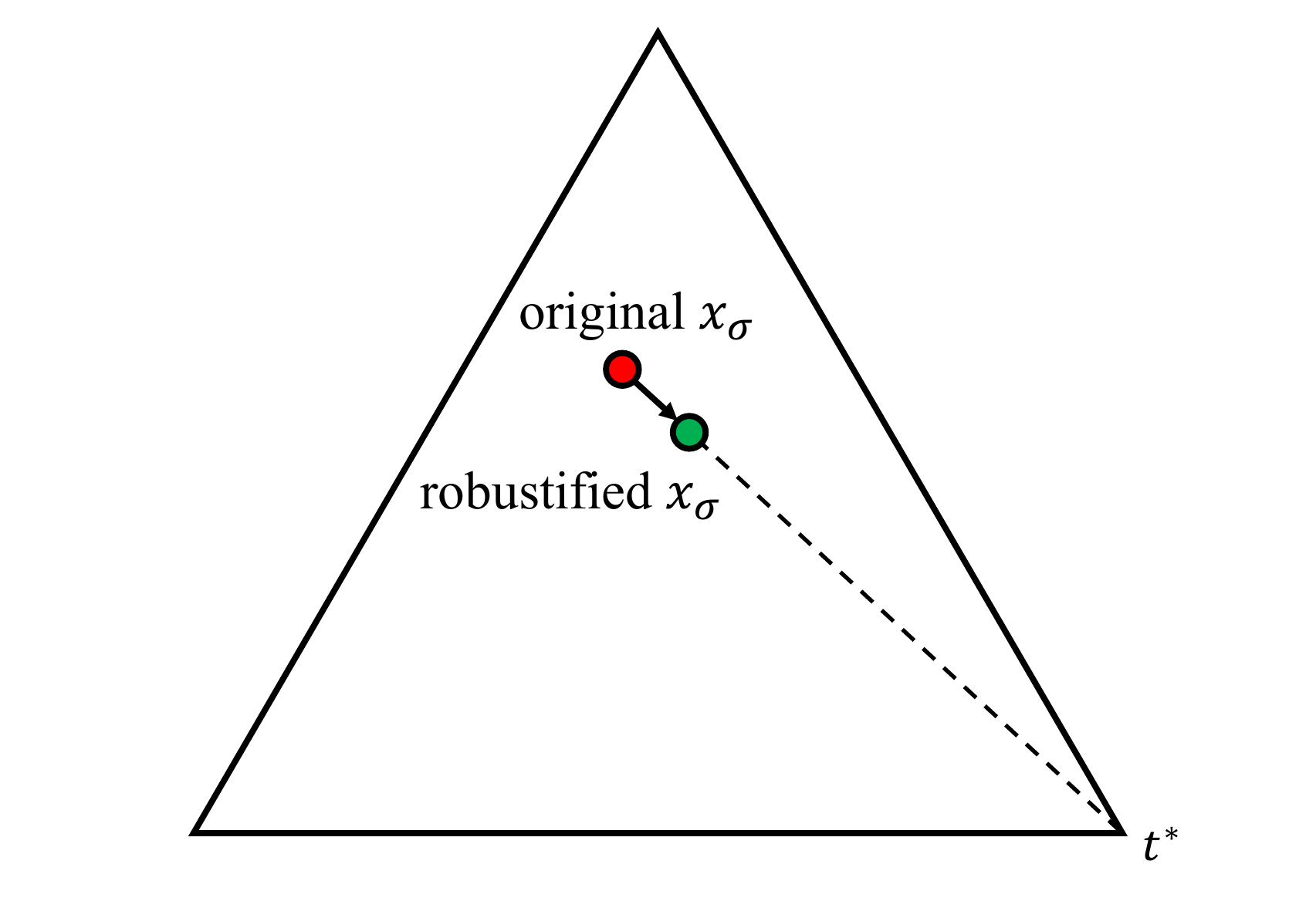}
        \caption{Robustify a significant segment}
        \label{fig:robustify-simple-1}
    \end{subfigure}
    \begin{subfigure}{.43\textwidth}
        \includegraphics[width=\textwidth]{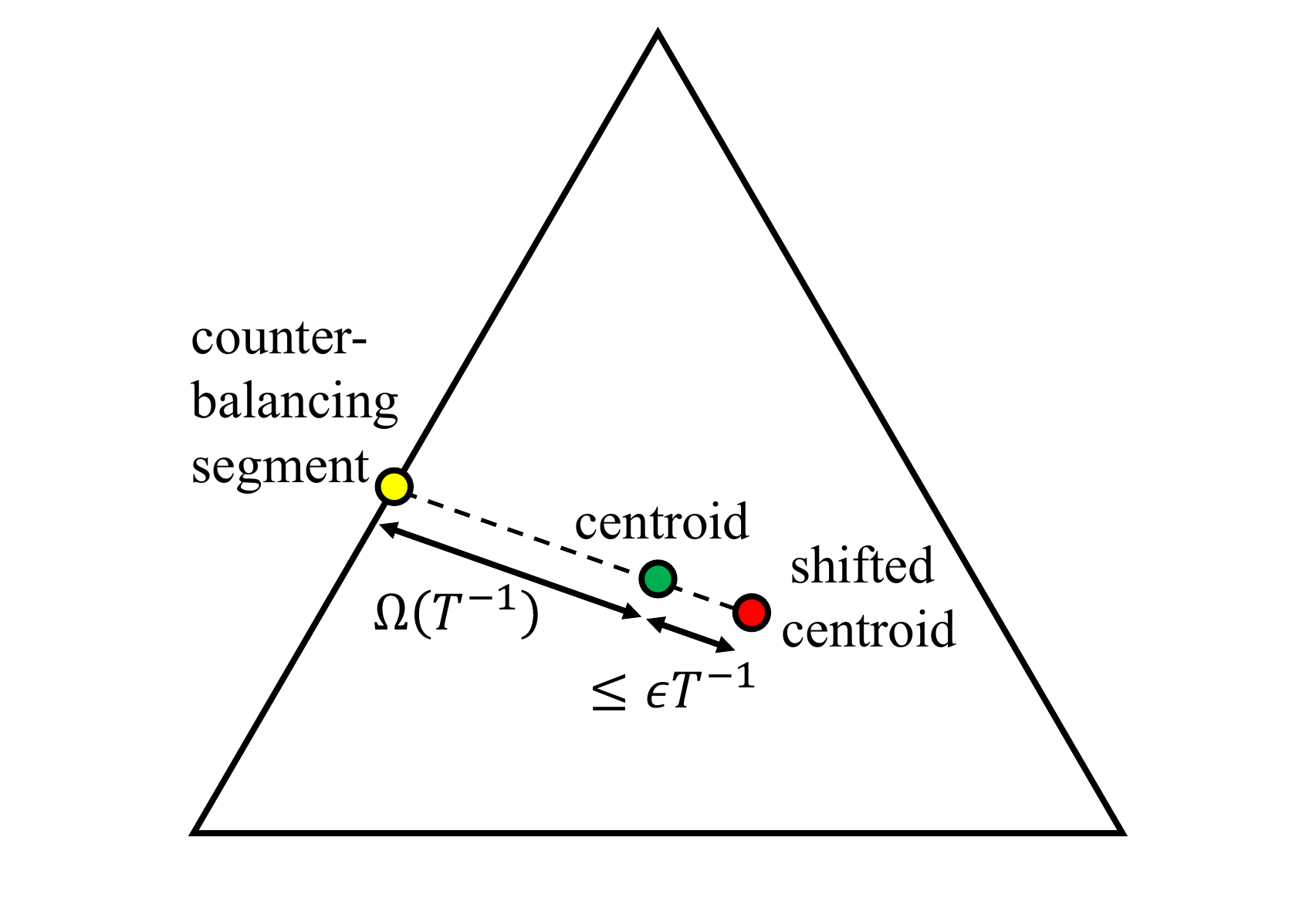}
        \caption{Counterbalance to restore the centroid}
        \label{fig:robustify-simple-2}
    \end{subfigure}
    \caption{Robustification. }
    \label{fig:robustify-simple}
\end{figure*}
\paragraph{Step 1:}
Robustify the significant segments one by one, ignoring the centroid constraint.
Any significant segment $\segment$ is represented by a point $\vec{\point}_\segment$ and weight $w_\segment = \mu(\vec{\point}_\segment)$, whose intended price is $p^*$.
In the simplex view, we want to move slightly  away from $\vec{\point}_\segment$ 
such that we end up far from all regions $\area_p$ 
where price $p$ gives a small consumer surplus.
How do we find which direction to move towards?  
(Since we are in high dimensions, we cannot rely on geometric intuition.)

The choice of direction to move towards relies on a structural result about the mixtures of MHR-like distributions stated as Lemma \ref{lem:robustify-find-type-main}.
The lemma promises that there exists a type $\type^*$ such that 
for the distribution $\valuedist(\type^*)$, 
for prices whose quantiles are less than that of $p^*$ by at least $\epsilon_I$, 
there is a revenue gap of $ \tfrac{\epsilon_S \typenum}{\epsilon_I}$. 
Once we prove existence, 
we can find such a type by enumerating over all types and checking if the property holds. 

In the simplex view, see Figure~\ref{fig:robustify-simple-1}.
We want to move $\vec{\point}_\segment$(red point) towards the vertex that corresponds to type $\type^*$(to the green point); 
we want to be at $(1-\tfrac {\epsilon_I} {T})\vec{\point}_\segment + \tfrac {\epsilon_I} {T} \vec{\point}_{\type^*} $.
To do this, decrease the probability of mapping each type to $\segment$ by an $1 - \frac{\epsilon_I}{\typenum}$ factor;
then, increase the probability of mapping $\type^*$ to $\segment$ additively by $\frac{w_\segment \epsilon_I}{\typenum}$ to restore the original weight.
Clearly, this satisfies the \emph{mixture preservation} condition.

Let $p$ be any price whose quantile is smaller 
than that of $p^*$ by at least $\epsilon_I$. 
The revenue gap between $p$ and $p^*$ for $\valuedist(\type^*)$ is $ \tfrac{\epsilon_S \typenum}{\epsilon_I}$, and 
we moved towards $\vec{\point}_{\type^*} $ by 
$\tfrac {\epsilon_I} {T} $, 
therefore 
the revenue gap between $p$ and $p^*$ 
for the mixture is at least $\epsilon_S.$
As a result, $p$ cannot be an $\epsilon_S$-optimal price in the resulting segment.
Thus, we have the \emph{robustness} condition.

\paragraph{Step 2:}
See Figure~\ref{fig:robustify-simple-2}.
We will add a counterbalancing segment(yellow) to restore the centroid.
After the first step, the centroid may be shifted from its intended location, i.e., the middle of the simplex, by up to $\frac{\epsilon_I}{\typenum}$.
Consider the line that crosses the intended centroid and the shifted one. 
Add a counterbalancing segment at its intersection with the boundary of the simplex on the opposite side of shifted centroid, with an appropriate weight that restores the centroid.
The weight is only $O(\epsilon_I)$ because the distance between the intended centroid and the counterbalancing segment, in fact, any point on the boundary of the simplex in general, is at least $\Omega(\frac{1}{T})$ by basic geometry.

Finally, the total weight may now exceed $1$ by up to $O(\epsilon_I)$.
Normalize the weights of all segments to restore a total weight of $1$.
It decreases the weights of the segments by at most $1 - O(\epsilon_I)$ and therefore satisfies the \emph{weight preservation} condition.

We will show the formal algorithm and analysis in Appendix~\ref{app:robustify-algorithm}.

\subsubsection{Structural Lemma and the Proof Sketch}
\label{sec:proof-robustify-find-type-main}

\begin{lemma}
	\label{lem:robustify-find-type-main}
	For any segment $\vec{\point} \in \simplex$ that is significant in the sense that $\E_{v \sim \valuedist(\vec{\point})} [ v ] \ge \epsilon_I$, and its corresponding monopoly price $p^*$, there is a type $\type^* \in \typeset$ such that for any price $p$ whose quantile w.r.t.\ $\valuedist(\vec{\point})$ is smaller than that of $p^*$ by $\epsilon_I$, we have:
	\[
	\Rev \big( p, \valuedist(\type^*) \big) < \Rev \big( p^*, \valuedist(\type^*) \big) - \tfrac{\epsilon_S \typenum}{\epsilon_I}
	~.
	\]
\end{lemma}
This is technically the most challenging part of the proof. 
We will show the full proof in Appendix~\ref{sec:robustify-find-type}.
By concavity of the revenue curves of MHR-like distributions, it suffices to consider the inequality when $p$ is the smallest price whose quantile w.r.t.\ $\valuedist(\vec{\point})$ is smaller than that of $p^*$ by $\epsilon_I$. 
Let this price be $\Bar{p}$.

Recall the plateau example in Figure~\ref{fig:plateau}.
From the picture, it is tempting to pick the type that corresponds to the ``right-most'' dotted revenue curve, as it has the desirable shape that the revenue rapidly decreases when the price increases from $p^*$.
There are several problems with this approach. 
First, the concept of ``right-most'' revenue curve is underdefined. 
Is it the one with the smallest monopoly price? 
Or the one with the largest monopoly sale probability? 
Second, even if we find a type whose revenue curve has the desirable shape, it still may not prove Lemma~\ref{lem:robustify-find-type-main}.
For example, it may not have a large enough optimal revenue in the first place and, thus, the RHS of the inequality in the lemma is negative.

Instead, we will prove the lemma by contradiction.
Intuitively, the contradiction will be that 
there is a type $t$ such that the revenue curve 
of $\valuedist(\type)$ has a large plateau; 
this is not possible for MHR-like distributions. 
The assumption to the contrary guarantees that
the revenue of $p^*$ is not much above that of any price between $\bar{p}$ and $p^*$. 
The following additional conditions 
formalize the `large plateau' notion: 
\begin{enumerate}
	\item Revenue between $\bar{p}$ and $p^*$ is not much higher: $\forall~ p' \in [p^*,\bar{p}]$ , 
	\[ \textstyle \Rev \big( p', \valuedist(\type)\big) - \Rev \big( p^*, \valuedist(\type)\big) \le \tilde{O} \big( \frac{\epsilon_S T^2}{\epsilon_I^2} \big)
	.\]
	%
	\item The plateau is high, i.e., revenue of $p^*$ is large:  $\Rev \big( p^*, \valuedist(\type)\big) \ge \Omega \big( \frac{\epsilon_I}{T} \big)$. 
	\item The plateau is wide, i.e., the quantiles of $\bar{p}$ and $p^*$ differ by at least $\Omega \big( \frac{\epsilon_I^2}{T} \big)$.
\end{enumerate}

The rest of the subsection assumes to the contrary that the inequality in the lemma fails to hold for all types.
Then, we use a probabilistic argument to show that there must be a type $\type$ 
such that the distribution $\valuedist(\type)$ satisfies the conditions mentioned above, and argue that these lead to a contradiction.

\paragraph{Probabilistic Argument.}

Consider sampling a type $\type$ according to the mixture induced by the segment.
We show that under the assumption to the contrary, the probability that 
condition 1 is violated is less than $O\big(\frac{\epsilon_I}{T}\big)$
(Lemma~\ref{lem:robustify-beat-union-bound}).
Further, we prove that the revenue of $p^*$
w.r.t. $\valuedist(\vec{\point})$, which is optimal for this distribution, is at least an $\Omega\big(\frac{1}{T}\big)$ fraction of the social welfare w.r.t.$\valuedist(\vec{\point})$
(Lemma~\ref{lem:robustify-optimal-SW-lower-bound}).
Then by the assumption that this segment is significant, this is at least $\Omega \big( \frac{\epsilon_I}{T} \big)$
(Lemma~\ref{lem:robustify-optimal-revenue-lower-bound}).
Then, by a Markov inequality type argument, there is at least an $\Omega\big(\frac{\epsilon_I}{T}\big)$ probability that Condition 2 is satisfied
(Lemma~\ref{lem:robustify-condition-two}).
Putting together, there is a positive chance that we sample a type $\type$ that satisfies the first two conditions.
Finally, we finish the argument by showing that the first two conditions actually imply the third one
(Lemma~\ref{lem:robustify-quantile-gap}).

\paragraph{Contradiction.}
The proof is a case by case analysis, so 
we present the bottleneck case which forces 
the choice we made in Eqn. \eqref{eqn:robustify-epsilons}. 
This is when the monopoly price $p(\type)$ of type $\type$ is smaller than both $\bar{p}$ and $p^*$. 
A complete proof that includes the other cases are deferred to 
Appendix~\ref{sec:robustify-find-type} (Lemma~\ref{lem:robustify-contradiction}).

The concavity, and strong concavity near monopoly price, of MHR-like distributions, 
along with the fact that both $\bar{p}$ and $p^*$ are larger than the monopoly price,
imply that the revenue gap is at least the revenue of $p^*$ times the square of the quantile gap between the prices.
Further by the second and third conditions above, of having large revenue and large quantile gap, the revenue gap between prices $\bar{p}$ and $p^*$ is at least:
\[
\Omega \bigg( \frac{\epsilon_I}{T} \bigg) \cdot \Omega \bigg( \frac{\epsilon_I^2}{T} \bigg)^2 = \Omega \bigg( \frac{\epsilon_I^5}{T^3} \bigg)
\]
This is greater than $ \frac{\epsilon_S T}{\epsilon_I}$ by our choice of $\epsilon_I \ge \epsilon_S^{\frac{1}{6}} T^{\frac{2}{3}} \log^{\frac{1}{6}} V$ 
in Eqn. \eqref{eqn:robustify-epsilons}.

\subsection{Proof of Theorem~\ref{thm:sample-complexity-mhr-like}: Project, Optimize, and Robustify}
\label{sec:sample-complexity-mhr-like-proof}

\begin{algorithm*}[t]
    \caption{Learn a (Robust) Segmentation from Samples}
    \label{alg:learn-segmentation-from-samples-main}
    \begin{enumerate}
        \item Construct empirical distributions $\empiricaldist(\type)$'s, $\type \in \typeset$, from samples.
        \item Find MHR-like empirical distributions $\widetilde{\empiricaldist}(\type)$'s, $\type \in \typeset$, such that the Kolmogorov-Smirnov distances between them and the corresponding empirical distributions are small:
                $
                d_{KS} \big( \widetilde{\empiricaldist}(\type), \empiricaldist(\type) \big) \le \epsilon_S
                $.
        \item Find optimal segmentation $(\segmentset^*, \segmentmap^*)$ w.r.t.\ MHR-like empirical distributions $\vec{\widetilde{\empiricaldist}}$.
        \item Construct the robustified segmentation $(\segmentset^r, \segmentmap^r)$ and return it.
    \end{enumerate}
\end{algorithm*}
Finally, we show how to use the robustification technique to design an algorithm, presented as Algorithm~\ref{alg:learn-segmentation-from-samples-main}, that learns a (robust) $O(\epsilon_I)$ segmentation in the sample complexity model.

\paragraph{Algorithm.}
%
Similar to the existing works on the sample complexity of mechanism design, the algorithm starts by constructing the empirical distributions.
Then, we \emph{project} them back to the space of MHR-like distributions w.r.t.\ the Kolmogorov-Smirnov distance $d_{KS}$, i.e., the maximum difference in the quantile of any value. 
The feasibility of this step comes from the fact that the true distributions are MHR-like and satisfy the inequality.
We explain in Appendix~\ref{app:find-nearby-mhr} how to compute in polynomial time approximate projections that relax the RHS of the inequality by a constant factor;
other constants in our analysis need to be changed accordingly but the bounds stay the same asymptotically. 
Further, we \emph{optimize} the segmentation according to the MHR-like empirical distributions.
Finally, we \emph{robustify} the resulting segmentation using Lemma~\ref{lem:robustify}.

\paragraph{Analysis.}
%
Note that for any type $\type$, the distances between the true distribution $\valuedist(\type)$ and the seller's belief $\valuedist_S(\type)$, between $\valuedist(\type)$ and the empirical distribution $\empiricaldist(\type)$, and between $\empiricaldist(\type)$ and the MHR-like empirical distribution $\widetilde{\empiricaldist}(\type)$ are bounded by $\epsilon_S$.
Hence, the distance between the seller's belief $\valuedist_S(\type)$ and the MHR-like distribution $\widetilde{\empiricaldist}(\type)$ is at most $3 \epsilon_S$ by the triangle inequality and, thus, the same conclusion holds replacing types with mixtures induced from the segments.
Therefore, for any segment in the segmentation chosen by the algorithm, the seller's monopoly price w.r.t.\ her beliefs is a $6\epsilon_S$-optimal price w.r.t.\ the MHR-like empirical distributions.
By Lemma~\ref{lem:robust-segmentation-performance-main}, the performance of the algorithm is an $O(\epsilon_I)$-approximation comparing with the optimal w.r.t.\ the MHR-like empirical distributions.

It remains to show that the optimal w.r.t.\ the MHR-like empirical distributions is an $O(\epsilon_I)$-approximation to the optimal w.r.t.\ the true distributions.
To do that, it suffices to find a good enough segmentation achieving this approximation.
For this we once again resort to Lemma~\ref{lem:robustify}, in particular, the existence of an $(\epsilon_I,\epsilon_S)$-robustified segmentation $(\segmentset^r, \segmentmap^r)$ for the optimal segmentation w.r.t.\ the true distributions $\valuedist(\type)$'s.
Note that the MHR-like empirical distributions $\widetilde{\empiricaldist}(\type)$'s are at most $O(\epsilon_S)$ away from the corresponding true distributions $\valuedist(\type)$'s, by triangle inequality.
Therefore, running $(\segmentset^r, \segmentmap^r)$ on the MHR-like distributions, with a seller who posts the monopoly price w.r.t.\ the MHR-like distributions, gives an $O(\epsilon_I)$-approximation by Lemma~\ref{lem:robust-segmentation-performance-main}.

\section{Bandit Model}
\label{sec:bandit}

In the bandit model, the intermediary interacts with the seller and the buyer repeatedly for $\samplenum$ rounds for some positive integer $\samplenum$, with the buyer's type-value pair freshly sampled in each round.
The goal is to maximize the cumulative objective during all $\samplenum$ rounds.
There are a large variations of models depending on the modeling assumptions. 
Next, we explain our choice.

\paragraph{Intermediary's Information:}
The intermediary does not know the value distributions at the beginning and, therefore, must learn such information through the interactions in order to find a good enough segmentation.
Further, the intermediary observes in each round only the purchase decision of the buyer, but not her value.
This is similar to the bandit feedback in online learning and hence the name of our model.
We remark that the alternative model where the intermediary can observe the values, which corresponds to full-information feedback in online learning, easily reduces to the sample complexity model as the intermediary may simply run the algorithm in the sample complexity model using the bids in previous rounds as the samples.

Since the intermediary can observe the buyer's type in each round, she can easily learn the type distribution through repeated interactions.
To simplify the discussions, we will omit this less interesting aspect of the problem and will assume that the type distribution is publicly known.
Following the treatment in previous models, we further assume that it is a uniform distribution.

\paragraph{Buyer's Behavior:}
We assume the buyer is myopic and, therefore, buys the item in each round if and only if the price posted is at most her value.
In other words, the buyer does not take into account that her behavior in the current round may influence how the intermediary and the seller acts in the future.
This challenge of non-myopic buyers was partly addressed in the online auction problem by \cite{HLW18}.
Their techniques, however, do not directly apply to our problem.

\paragraph{Seller's Information:}
We assume that, like the intermediary, the seller does not have any information about the value distributions of the buyer at the beginning, and must learn such information through bandit feedback. 
With this assumption, we will investigate how to encourage the seller to explore on the intermediary's behalf.
What makes it challenging is that the seller's objective (revenue) and the intermediary's objective (e.g., social welfare) may not be aligned.


\paragraph{Seller's Behavior:}
Any algorithm by the intermediary must rely on some assumptions on the seller's behavior to get a non-trivial performance guarantee.
Informally, we need the seller to pick an (approximately) optimal price in terms of revenue when there is enough information for finding one; 
there is not much we can do if the seller simply ignores any information and picks prices randomly.
On the other hand, we also need the seller to explore at a reasonable rate in order to learn the value distributions.
If the seller could have other sources of information 
which allow him to estimate the distribution accurately, he may severely limit his exploration on prices whose confidence intervals suggest high-potential (and high uncertainty).
Our assumption must disallow such strategies and ensure that  the seller learns the distributions
only via observing the buyer's actions. 
\emph{What are the mildest behavioral assumptions (on the seller) that allow the intermediary to have a non-trivial guarantee in bandit model?}

Note that the seller herself faces an online learning problem with bandit feedback.
Our model is driven by the exploration-exploitation dilemma in her viewpoint.
First, we introduce the upper confidence bound (UCB) and the lower confidence bound (LCB) of the quantile of any value $v$ and any type $\type$ given past observations in the form of (type, price, purchase decision)-tuples. 
\begin{enumerate}
    \item 
    For any value $v \in \valueset$ and any type $\type \in \typeset$, suppose there are $\samplenum(v,\type)$ past observations with type $\type$ and price $v$, among which the buyer purchases the item in $\samplenum^+(v,\type)$ observations.
    Then, for some constant $C > 0$ that depends on the desired confidence level, let:
    \begin{align*}
        \widetilde{U}(v,\type) 
        &
        = \tfrac{\samplenum^+(v,\type)}{\samplenum(v,\type)} + \sqrt{\tfrac{C}{\samplenum(v,\type)}} \\
        \widetilde{L}(v,\type)
        &
        = \tfrac{\samplenum^+(v,\type)}{\samplenum(v,\type)} - \sqrt{\tfrac{C}{\samplenum(v,\type)}}
    \end{align*}
    \item
    Noting that quantiles are monotone, we define the UCB and LCB as follows:
    %
    \begin{align*}
        U(v,\type)
        &
        = \min_{v' \le v} \widetilde{U}(v',\type) &
        L(v,\type)
        &
        = \max_{v' \ge v} \widetilde{L}(v',\type) 
    \end{align*}
    \item
    This further induces the UCB and LCB of the quantile of value $v$ w.r.t.\ each segment $\vec{\point} \in \simplex$:
    \begin{align*}
        U(v,\vec{\point}) 
        &
        = \sum_{\type\in[\typenum]} \point_\type \cdot U(v,\type) \\
        L(v,\vec{\point})
        &
        = \sum_{\type\in[\typenum]} \point_\type \cdot L(v,\type)
    \end{align*}
\end{enumerate}

For some target average regret $0 < \epsilon_S < 1$ of the seller, we say that she \textbf{exploits} in a round if the segment $\vec{\point}$ and her price $p$ satisfy that:
\[
    U(p,\vec{x}) \ge \max_{p' \in \valueset} U(p',\vec{x}) - \epsilon_S
    ~.
\]
Otherwise, we say that she \textbf{explores}.
We assume that the seller is an $\epsilon_S$-canonical learner in the sense that she exploits in all but at most an $\epsilon_S$ fraction of the rounds.

Among others, we give two example algorithms that satisfy this definition.
First, the Upper-Confidence-Bound (UCB) algorithm satisfies this with $\epsilon_S = 0$.
Further, consider the following simple Explore-then-Commit (ETC) algorithm, with $\epsilon_S = \tilde{O} \big( m^{-\frac{1}{3}} \typenum \valuenum \big)$.
The seller explores in the first $\epsilon_S \samplenum = \tilde{\Omega} \big( \typenum \valuenum \epsilon_S^{-2} \big)$ rounds by posting random prices.
Since a price-type pair shows up with probability $\frac{1}{\typenum \valuenum}$ in each of these rounds, she learns the quantile of every price $p \in \valueset$ w.r.t.\ the value distribution $\valuedist(\type)$ of every $\type \in \typeset$ up to an additive error of $\epsilon_S$ .
Then, in the remaining rounds, she can pick any price that is not obviously suboptimal in the sense that its UCB is smaller than the LCB of some other price.

\subsection{Algorithm}

In each round, the algorithm seeks to place the intermediary in a win-win situation by maintaining a set of optimistic hypothetical value distributions for the types, together with a robustified version of the optimal segmentation w.r.t.\ the hypothetical distributions.
If the seller indeed posts a price that is consistent with our optimistic hypothesis, we use the analysis from the previous section (Thm.~\ref{thm:sample-complexity-mhr-like}) to show that the objective in this round is close to optimal.
Otherwise, if the seller posts a price that is inconsistent with our optimistic hypothesis, we argue that there must be a sufficiently large gap between its UCB and LCB and, thus, the intermediary gets some useful new information.

There is a caveat, however, when the algorithm constructs the robustified segmentation:
it needs to replace $\epsilon_S$ with some slightly larger parameter $\epsilon_M$.
In particular, let $\epsilon_M$ be such that if we define $\epsilon_I$ using Eqn.~\eqref{eqn:robustify-epsilons}, replacing $\epsilon_S$ with $\epsilon_M$, we have: $\epsilon_I m = \typenum \valuenum \epsilon_M^{-3}$.
Solving it gives that $\epsilon_M = \Theta \big( m^{-\frac{6}{19}} \textrm{poly}(\typenum, \valuenum) \big)$ and $\epsilon_I = \Theta \big( m^{-\frac{1}{19}} \textrm{poly}(\typenum, \valuenum) \big)$. 

We show that for $\epsilon_S$-canonical learners with a sufficiently small $\epsilon_S$, which is satisfied by both aforementioned examples, we can get sublinear regret.

\begin{algorithm*}[t]
    \caption{Segmentation algorithm in the bandit model (in each round)}
    \label{alg:bandit}
    \begin{enumerate}
        \item 
        Let $\valuedistset(t)$ be the set of MHR-like distributions with support $\valueset$ such that the quantile of each value $v \in \valueset$ is between the corresponding UCB and LCB.
        \item
        Let $\valuedist^*(t) \in \valuedistset(t)$, $t \in \typeset$, be such that $\Opt(\vec{\valuedist}^*)$ is maximized.
        \item
        Find $(\segmentset^*, \segmentmap^*)$, the optimal segmentation w.r.t.\ $\vec{\valuedist}^*$, via the algorithm in Section~\ref{sec:bayesian}.
        \item 
        Construct $(\segmentmap^r, \segmentset^r)$, a robustified version of $(\segmentset^*, \segmentmap^*)$, via Algorithm~\ref{alg:robustify-segmentation}, using $\epsilon_M$ in place of $\epsilon_S$.
    \end{enumerate}
\end{algorithm*}


\begin{theorem}
    \label{thm:bandit}
    Algorithm~\ref{alg:bandit} gets at least $\Opt - O \big(m^{-\frac{1}{19}} \cdot \poly(\valuenum, \typenum) \big)$ per round on average, provided that the seller is an $\epsilon_S$-canonical learner with $\epsilon_S \le O(\epsilon_M) = O \big( m^{-\frac{6}{19}} \cdot \poly(\valuenum, \typenum) \big)$. 
\end{theorem}

If the seller explores in a round such that the corresponding UCB and LCB differs by not only $\epsilon_S$ but by at least $\epsilon_M$, we call it a \emph{major exploration}.
We first upper bound the number of rounds that involve such major explorations in the following lemma.

\begin{lemma}
    \label{lem:exporation-bound}
    The expected number of rounds that are major explorations is at most $\tilde{O}(\typenum \valuenum \epsilon_M^{-3})$.
\end{lemma}

\begin{proof}
    Every time that the seller makes a major exploration on some price $p$ in a round, say, in response to a segment represented by a point $\vec{\point} \in \simplex$, the gap between the UCB and the LCB is at least $\epsilon_M$.
    Then, the expected gap between $\widetilde{U}(p, \type)$ and $\widetilde{L}(p, \type)$ is also at least $\epsilon_M$ when type $\type$ is sampled according to $\vec{\point}$.
    This implies that, with probability at least $\frac{\epsilon_M}{2}$, the realized type $\type$ actually has a gap of at least $\frac{\epsilon_M}{2}$ between $\widetilde{U}(p, \type)$ and $\widetilde{L}(p, \type)$.
    Note that for any type $\type$, and any price $p$, this cannot happen by more than $O \big(\epsilon_M^{-2}\big)$ times by the definitions of $\widetilde{U}(p, \type)$ and $\widetilde{L}(p, \type)$.
    Hence, the expected number of times that the seller explores cannot exceed $O \big( \typenum \valuenum \epsilon_M^{-3} \big)$ times.
\end{proof}

We now prove Theorem~\ref{thm:bandit}.
\paragraph{Case 1: Seller picks an undesirable price.}
Suppose the seller fails to pick a price that is at least $14\epsilon_M$-optimal w.r.t.\ the optimistically chosen distributions $\vec{\valuedist}^*$.
Instead, she chooses a price $p$.
Then, either she is not exploiting in the sense of the definition in Section~\ref{sec:bandit}, which cannot happen in more than an $\epsilon_S$ fraction of the rounds, or the UCB of $p$ is at least the maximum UCB among all prices less $\epsilon_S$.
This is larger than the expected revenue induced from the optimistically chosen distribution $\vec{\valuedist}^*$ by at least $14 \epsilon_M - \epsilon_S > \epsilon_M$, by the assumption that $p$ is not $14\epsilon_M$-optimal.
Note that the latter is weakly larger than the LCB.
Hence, we conclude that the UCB and LCB differs by at least $\epsilon_M$, which means that this is a major exploration that cannot happen in more than an $m^{-1} \cdot O \big(\typenum \valuenum \epsilon_M^{-3}\big) = O(\epsilon_I)$ fraction of the rounds (Lemma~\ref{lem:exporation-bound}).

\paragraph{Case 2: Sale probability is lower than expected.}
Next, consider the case when the seller picks a price that is indeed $14\epsilon_M$-optimal w.r.t.\ the optimistically chosen distributions $\vec{\valuedist}^*$, but the sale probability of the price given by the true distributions is smaller than by $\vec{\valuedist}^*$ by more than $\epsilon_I$.
In this case, note that both sale probabilities are bounded between the UCB and LCB; we again conclude that there is a gap of at least $\epsilon_I > \epsilon_M$ between them.
Hence, this is a major exploration which cannot happen in more than $O(\epsilon_I)$ fraction of the rounds.

\paragraph{Case 3: Everything goes as expected.}
Finally, consider the good case, when the seller indeed picks a price that is at least $14\epsilon_M$-optimal least 4\ w.r.t.\ the optimistically chosen distributions $\vec{\valuedist}^*$, and that the sale probability of the price given by the true distributions is at least that by $\vec{\valuedist}^*$ less $\epsilon_I$.
Then, by Lemma~\ref{lem:robust-segmentation-performance-main}, we get that the expected objective in this round is at least $\Opt - O(\epsilon_I)$.

Since the first two cases cannot happen in more than an $O(\epsilon_I)$ fraction of the rounds, the bound stated in Theorem~\ref{thm:bandit} follows.



\section{Further Related Work}
\label{sec:related}


 
The problem of price discrimination is highly related to \emph{screening} in games of asymmetric information, pioneered by \cite{Spe73, CL00}, where the less informed player moves first in hopes of combating adverse selection. In our setting, the seller wishes to screen buyers by charging different prices depending on the buyer's value.  The intermediary's segmentation allows the seller to screen more effectively.  The intermediary themselves face a signaling problem, as their choice of segmentation is effectively a signaling scheme to the seller.  

As such, our work is related to the broad literature on signaling and information design, 
 where a mediator designs the information structures available to the players in a game \cite{BM16}. 
A special case of this 
 is known as \emph{Bayesian persuasion} \cite{KG11,DX16}:  
 an informed sender (here the intermediary) sends a signal about the state of the world to a receiver (here the seller), who must take an action that determines the payoff of both parties.  The goals of the sender and receiver may not be aligned, so the sender must choose a signaling scheme such that the receiver's best response still yields high payoff for the sender.
See \cite{Dug17,BM18} for surveys on these topics. 

Our results for online learning are also related to work on iteratively learning prices~\cite{BH05,MM14,CGM15,PPV16,BDHN17,HLW18}. Both lines of work consider the seller's problem of incentive design or learning, but do not have an intermediary or market segmentation component. 
 Our model is also somewhat related to the literature on \emph{dynamic mechanism design}, which considers the incentive guarantees of multi-round auctions where the same bidders may participate in multiple rounds. \cite{BV10, KLN13,PST14} gave truthful dynamic mechanisms for maximizing social welfare and revenue. 
 
Our results  are related to recent work on incentivizing exploration in a bandit model \cite{FKKK14,MSS15,MSSW16}.  These papers typically model a myopic decision-maker in each round, and an informed non-myopic principle who can influence the decision-maker to explore rather than exploit. In our setting, the seller is myopic decision-maker who sets prices, and the intermediary can influence that decision by changing the segmentation.  The previous results do not directly apply to our setting, as an action corresponds to setting a price in the observed segment.  Hence there are exponentially many actions, so one should not hope for polynomial run time or good regret guarantees by directly applying those results. Additionally, the intermediary chooses the segmentation but the observed segment is chosen randomly, so the intermediary cannot force the seller to play any particular action.

\section{Future Work}
\label{sec:future}

We view our results as initiating a new line of work on algorithmic price discrimination under partial information. We believe there are many promising open problems left to be explored in this direction, and hope this paper inspires future work under other informational models and market environments. We now present some of the most interesting directions for future work.

\paragraph{Competitive Markets.}
This paper and \cite{bergemann2015limits} consider a monopolist seller, 
 which is a good fit for something like an ad exchange. 
In many online marketplaces the sellers are in a competitive 
 rather than a monopolistic setting. 
The products are differentiated so sellers can exert some pricing power, 
 which still incurs deadweight loss. 
It would be very interesting to extend this theory of price discrimination 
 to such competitive markets.

\paragraph{Strategic Buyers.}
When a seller uses past buyer behavior
 in the form of auction bids
 or purchase decisions
 to decide future prices, 
 and a buyer has repeated interactions with such a seller,
 the buyer may be incentivized to strategize. 
Even if each interaction in isolation is strategyproof, 
 the buyer may forgo winning an earlier auction 
 in order to get a lower price in the future. 
When each buyer represents an insignificant fraction of the entire market, 
 techniques from differential privacy can address this issue \cite{HLW18}. 
In this paper we ignore this strategic aspect and 
 assume that buyers are myopic. 
It would be very interesting to get results analogous to \cite{HLW18}, 
 since their techniques do not readily extend to our model.

\paragraph{Worst Case Model.}  
In online learning, 
 even for arbitrary sequence of inputs, 
 we can often get a regret guarantee matching 
 that for an i.i.d.\ input sequence. 
In particular this is true for a monopolistic seller 
 learning an optimal price or an optimal auction \cite{BDHN17,BH05}. 
It is tempting to conjecture that the same holds for our setting as well, 
 but we run into difficulties even modeling the problem. 
How does the seller behave in such a scenario? 
In this paper, we modeled the seller behavior 
 based on the underlying distribution. 
Defining a seller behavior model in the absence of such a distribution that is both reasonably broad 
 and allows regret guarantees in the worst case setting 
 is an interesting challenge.





\bibliographystyle{plainnat}
\bibliography{reference}

\appendix

\section{Further Examples}
\label{sec:example}


\subsection{Benefit of Segmentation}
We now reproduce an example from \cite{bergemann2015limits} 
that shows how a segmentation can eliminate deadweight loss.

\begin{Example}
The value set $\valueset = \{1,2,3\}$ is identical to the type set $\typeset$, and each distribution $\valuedist(\type)$ is a pointmass at $t$. The prior distribution over values/types is uniform. The monopoly reserve of the uniform prior distribution is  $2$.  When the seller does not segment the market and posts the monopoly reserve price, revenue is $4/3$, consumer surplus is $1/3$, and the deadweight loss is $1/3$. 
\end{Example}

Consider instead the following segmentation with $\segmentset = \{\sigma_1,\sigma_2,\sigma_3\}$ where $\segmentmap(1) = (1,0,0)$; $\segmentmap(2) = (1/3,1/6,1/2)$; and $\segmentmap(3) = (2/3,1/3,0)$.  Recall that $\segmentmap(t)$ is the distribution of signals sent by the intermediary upon observing type $t$.  This signaling scheme generates three market segments, one corresponding to each of $\sigma_1,\sigma_2,\sigma_3$. The seller can compute the conditional distribution of values within each segment, and will post the monopoly price for that distribution. Within segment $\sigma_1$, the distribution of values is: $x_1$ with probability 1/2, $x_2$ with probability 1/6, and $x_3$ with probability 1/3. This happens to be the equal revenue distribution on values $\{1,2,3\}$.  Within segment $\sigma_2$, only buyers of types 2 or 3 will be present, since a buyer of type $1$ will never be mapped into this segment. The conditional distribution of values in this segment is: $x_2$ with probability 1/3 and $x_3$ with probability 2/3, which also happens to be the equal revenue distribution on values $\{2,3\}$. Type $2$ is the only type with positive probability of being mapped to segment $\sigma_3$, so the value distribution in segment $\sigma_3$ is a point mass on value 2.

Since each market segment has an equal revenue distribution of values, the seller can maximize his profit by posting any price in the support of that distribution. For the sake of this example, we will assume the seller breaks ties by posting the lowest optimal price.\footnote{ This can be strictly enforced using arbitrarily small perturbations.} In market segment $\sigma_1$, the seller will post price $p=1$, which will generate revenue of $1$ and consumer surplus of $(1/6)\times 1 + (1/3)\times 2 = 5/6$. In $\sigma_2$, the seller will post price $p=2$, which will generate revenue of $2$ and consumer surplus of $(2/3)\times 1 = 2/3$.  Finally, in market segment $\sigma_3$, the seller will post price $p=2$, which will generate revenue of $2$ and consumer surplus of $0$.  

Given the prior distribution of values, one can also compute the probability of the intermediary generating each market segment.  In this case, segment $\sigma_1$ is drawn with probability $2/3$, segment $\sigma_2$ is drawn with probability $1/6$, and segment $\sigma_3$ is drawn with probability $1/6$. Combining these, we see that the expected revenue of this segmentation is $(2/3)\times 1 + (1/6)\times 2 + (1/6)\times 2 = 4/3$, and the consumer surplus is $(2/3)\times (5/6) + (1/6)\times (2/3) + (1/6)\times 0 = 2/3$. This segmentation has zero deadweight loss, which means that full economic efficiency is achieved.


\subsection{Noisy Types}
\label{sec:noisy-types}

Now consider the above example with noise.
This example has been briefly discussed in Section~\ref{sec:bayesian}.
This subsection includes more details.

\begin{Example}
\label{example:noisy-types}
Given a noise level of $1-z$, 
each type $\type \in \typeset$ corresponds to the distribution 
$\valuedist(\type)$ given by 
$$
\Pr_{\val \sim \valuedist(t)} [\val] = \twopartdef{ z }{\text{ if } \val = \type,  } { \tfrac{1-z}{2}}{\text{otherwise.}}
$$
\end{Example}

Note that when $z=1/3$, all $\valuedist(t)$ equal the uniform prior, and no further (non-trivial) segmentation is possible. At the other extreme, when $z=1$ each $\valuedist(t)$ is a pointmass at $t$, which corresponds to Example 1. For $z=0.49$, it turns out that market segmentation cannot help at all. 
Any segmentation will result in the seller always setting the monopoly reserve price of $2$, and the result is the same as no segmentation! As in Example 1, revenue from no segmentation is $4/3$, consumer surplus is $1/3$, and the deadweight loss is $1/3$. 

On the other hand, when $z=0.8$ segmentation is possible, but it is no longer possible to achieve full economic surplus due to the noisy types.  If the intermediary implemented the same segment map as in Example 1, the resulting conditional value distributions in each segment would no longer be equal revenue distributions because types are no longer perfectly correlated with values.  If the intermediary perfectly segments the market by types using the deterministic segment map $\segmentmap(t)=\sigma_t$, this will result in revenue $2-\tfrac{4(1-z)}{3}=1.7333$, consumer surplus $\tfrac{2(1-z)}{3}=0.1333$, and deadweight loss $\tfrac{2(1-z)}{3}=0.1333$.  Note that under the noiseless setting of Example 1, this segmentation would have allowed the seller to perfectly price discriminate, resulting in revenue $2$, consumer surplus $0$, and deadweight loss $0$.  These changes in economic outcome are a direct result of the fact that types are only a noisy signal of the buyer's value.  For example in market segment $\sigma_3$, the seller will still set monopoly price $p=3$, but only a $(1-z)$-fraction of the segment will have value $3$ and purchase the item.

\subsection{Impossibility of \cite{bergemann2015limits} Style Characterization}
Recall the simplex view and Figure~\ref{fig:simplex}.
In the setting of \cite{bergemann2015limits}, 
the segmentation only consists of points in $X_{p^*}$ 
where $p^*$ is the monopoly price of the prior distribution 
(which is the blue region $X_2$ in Figure \ref{fig:simplex1}). 
As mentioned earlier, \cite{bergemann2015limits} assume that ties for monopoly price are broken in favor of the lowest price.
This implies that even though the market is now segmented, 
the price in each segment is at most $p^*$.
In other words, segmentation can only lower prices!  

One could hope for a similar characterization even in the noisy signal case: 
after all we can still write $\centroid$ as a convex combination of 
points in the intersection of $X_{p^*}$ (the blue region) and 
$\simplex(\typeset)$ (the convex hull of the type distributions, depicted as the triangles in the interior of the simplex in Figure \ref{fig:simplex1}),
as can be seen in Figure \ref{fig:simplex2}. 
Unfortunately we show that this is not without loss of generality
in the following example. 
In particular, the example shows that restricting to such segments 
may lead to a strictly smaller social welfare than otherwise.

\begin{figure*}
	\centering
	\begin{subfigure} {.4\textwidth}
		\includegraphics[width=\textwidth]{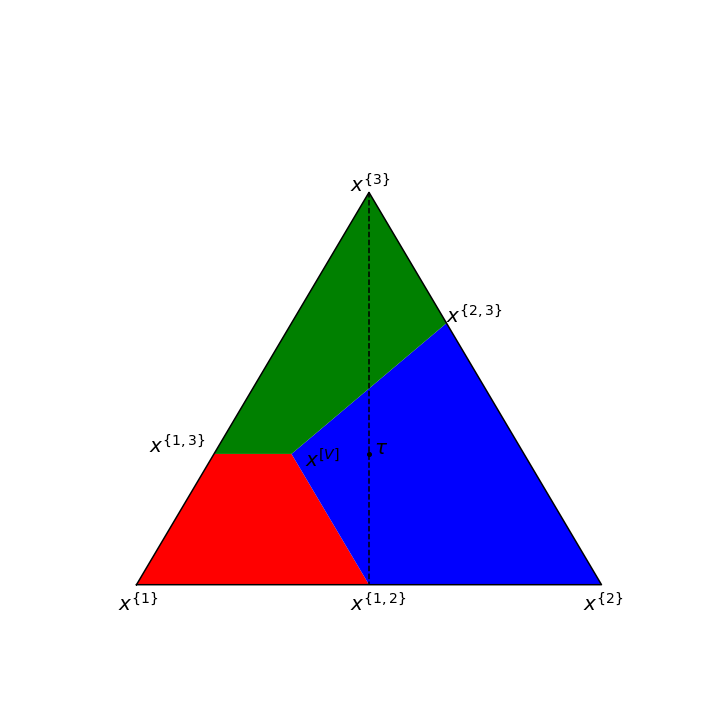}
		\caption{The two types correspond to the points $x^{\{1,2\}}$ and $x^{\{3\}}$. 
			The line joining them is the convex hull $\simplex(\typeset)$ from which we can pick our segments. }
		\label{fig:simplex3} 
	\end{subfigure}
	\begin{subfigure}{.47\textwidth}
		\centering
		\includegraphics[width=\textwidth]{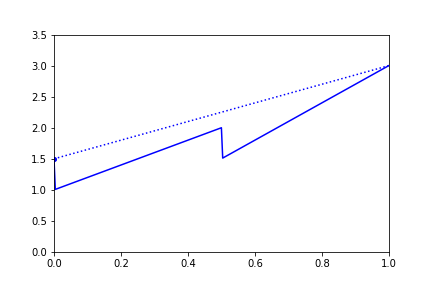}
		\caption{The solid line represents the social welfare as you move along the line from $x^{\{1,2\}}$ to $x^{\{3\}}$. The dotted line corresponds to the social welfare of the corresponding convex combination of the two end points}
		\label{fig:simplex4} 
	\end{subfigure}	
	\caption{Example in Appendix A.3}
\end{figure*}

\begin{Example}
In this example, there are two types,  
corresponding to the points $x^{\{1,2\}}$ and $x^{\{3\}}$,
i.e., type 1 is the uniform (equal revenue) distribution over ${\{1,2\}}$, and 
type 2 is the point mass on 3. 
The distribution over types is still uniform, 
which once again gives the same prior distribution $\centroid$ as before. 
\end{Example}

The difference now is that we can only pick points from the line joining these two points, as depicted in Figure \ref{fig:simplex3}. 
Figure \ref{fig:simplex4} shows the social welfare obtained 
as we move along this line. 
The point mass on the left corresponds to  $x^{\{1,2\}}$, 
for which we assume that the seller picks a price of 1. 
The left segment corresponds to the blue region $X_2$, 
and the right to the green region $X_3$.
Using only the blue region means using either 
the point mass on the left or 
the left segment. 
The dotted line in Figure \ref{fig:simplex4} shows the social welfare obtained from 
taking a convex combination of the end points. This corresponds to the segmentation where the intermediary 
simply reveals the type that he observes. 
As can be seen from the figure, this is strictly better than restricting to points in $X_2$.

Nonetheless, we can add this as an additional constraint
if so desired (at some loss in the objective). 
Our algorithm extends to handle this easily: 
just skip iterating over $X_p$ for prices $p > p^*$.

\subsection{Unbounded Sample Complexity in the General Case} 
So far we have assumed that the prior distribution is given to us as input 
 and is common knowledge to all players. 
How does this happen? 
What if you only have samples from the distribution? 
How many samples do you need in order to get within 
 an $\epsilon$ of the optimum? 
These questions have been studied under `sample complexity of auction design' 
 quite intensively in the last few years. 
(See Section \ref{sec:related} for  details on this line of work.)
Only recently, the optimal sample complexity of single item auctions 
 has been resolved \cite{GHZ19}. 
 In this paper, we consider the sample complexity of market segmentation.

Unlike auction design, here the seller sets the price
 to maximize revenue but 
 the intermediary's objective may be something different, 
 such as consumer surplus. 
For an equal revenue distribution, statistically
 the samples will indicate that the high price is revenue-optimal 
 with a significant probability. 
This is still true even for distributions that are ``close to'' equal revenue where the low price is strictly revenue-optimal. Higher prices correspond to lower social welfare, since the buyer is less likely to purchase the good;
thus the segmentation based on the samples could have 
 a much smaller consumer surplus compared to 
 the optimal segmentation for the distribution. 
This is particularly problematic because as saw earlier, 
 the optimal segmentation only picks segments with equal revenue distributions. 
(Recall that the vertices of the colored regions correspond to 
 equal revenue distributions.)
We demonstrate this via the following example. 

\begin{Example}
Consider the distribution on $V = \{ 1,2\}$ where 
 the probability of seeing $1$ is $0.5 + \delta,$
 for some very small $\delta>0$.
We will make the segmentation problem trivial: 
 there is only one type, i.e., the intermediary observes no signal. 
The monopoly price is 1, and consumer surplus is therefore $0.5 - \delta$; 
 this is trivially the optimal consumer surplus 
 we can obtain via segmentation.  
\end{Example}

When drawing multiple samples from this distribution, 
 there is a constant probability of seeing see more 2s than 1s. 
In this case, the seller sets a monopoly price of 2, 
 leading to a consumer surplus of 0. 
For any bounded function $f$ of $\epsilon$, 
 we can set $\delta$ small enough such that 
 with $f(\epsilon)$ samples, 
 this happens with probability $> \epsilon$; 
 we cannot therefore hope to get to within $\epsilon$ of the optimum. 
This example shows that we need a stronger assumption on the input distributions, 
 as compared to those for single item auctions. 
The standard assumptions there are regularity and boundedness, 
 both of which are satisfied by the distribution in the example above.

In the above example, we assumed that 
 the seller sets the monopoly price on the empirical distribution 
 on the samples. 
This is not necessarily an accurate assumption. 
If the seller follows the literature on the 
 sample complexity of single item auctions, 
 he should consider a robust or a regularized version of the empirical distribution. 
The seller might also have some additional sources of information 
 that allow him get an even more accurate estimate. 
We make a mild assumption on the seller behavior: 
 that his beliefs are at least as accurate as 
 the intermediary's estimate from the samples.
For a formal definition and more detailed discussion, 
 see Section \ref{sec:sample}.

 
Given the discussion so far, it may be tempting to assume 
 that the type distributions are far from the boundaries of $\area_\val$s. 
This is too strong when you consider larger value ranges: 
 two prices, say $p$ and $p+1$, may have almost identical quantiles and, thus, revenues 
 (which means the distribution close to the boundary 
 between $\area_p$ and $\area_{p+1}$), 
 but this is not a problem if they both give similar consumer surpluses. 
We make a milder assumption, which is  
 a discrete version of the monotone hazard rate (MHR) assumption.\footnote{
	MHR distributions have non-decreasing hazard rate $\tfrac{f(x)}{1-F(x)}$, where $f$ and $F$ are the pdf and the cdf of the distribution respectively.}
Here, the naive way to generalize MHR to discrete distributions is
 to use the functional form. 
This doesn't seem to work; 
we instead require the distribution to have some of the properties
 that continuous MHR distributions are known to satisfy. 
Once again, a detailed discussion on this with formal statements 
 are in Section \ref{sec:sample}. 

\section{Missing Proofs in Section~\ref{sec:bayesian}}
\label{app:bayesian-proof}

\noindent\textbf{Proof of Lemma~\ref{thm:simplex-view}}
{\em From Segmentations to Distributions.}
Given any segmentation $(\segmentset, \segmentmap)$, consider the  set $\area$ that is the union of the following points.
For any segment $\segment \in \segmentset$, let there be a point $\vec{\point}(\segment) \in \simplex(\typeset)$ such that:
\[
\point_\type(\segment) = \frac{\Pr_{\typedist}[t] \cdot \Pr[\segmentmap(t) = \segment]}{\Pr_{\segmentdist}[\segment]}
~.
\]
The probability distribution over $X$  is defined as 
$
\measure(\vec{\point}(\segment)):= \Pr_{\segmentdist}[\segment]
$.
It is easy to verify that this is indeed a probability distribution and that it satisfies the expectation constraint \eqref{eqn:centroid_constraint}. 

{\em From Distributions to Segmentations.}
Given any $(\area,\measure)$ that satisfies (\ref{eqn:centroid_constraint}), 
consider the following segmentation $(\segmentset, \segmentmap)$.
Let $\segmentset $ be such that there is some bijection between $\segmentset$ and $\area$, where $\segment \in \segmentset$ is 
uniquely mapped to $\vec{\point}(\segment) \in \area$.
For any type $\type$, let $\segmentmap(\type)$ follow a distribution given by 
\[
\Pr[\segmentmap(\type) = \segment ] = \frac{\point_\type (\segment) \measure(\vec{\point})}{\Pr_{\typedist}[\type]}
~.
\]
The pair $(\area, \measure)$ satisfies \eqref{eqn:centroid_constraint}  
 implies that $\segmentmap(\type)$ are valid probability distributions: 
\[
\forall~ \type \in \typeset, 
\sum_{\segment \in \segmentset} \Pr[\segmentmap(\type) = \segment] = \sum_{\segment \in \segmentset}\frac{ \point_\type (\segment) \measure(\vec{\point})}{\Pr_{\typedist}[\type]} = 1
~.
\]

\noindent\textbf{Proof of Lemma~\ref{lem:bayesian-finite-support}}
    Consider any optimal segmentation with an arbitrary, or even unbounded number of segments.
    We show how to transform it into a segmentation with the same objective and yet having at most one segment within each region $\area_p$, replacing the distribution conditioned on $\vec{x} \in \area_p$ by its expectation.

    Consider a specific area $\area_p$.
    Let $\mu_p = \Pr_\mu \big[ \vec{\point} \in \area_p \big]$ be the probability of realizing a segment in $\area_p$.
    Let $\vec{\point}_p = \E_\mu \big[ \vec{x} \,|\, \vec{x} \in \area_p \big]$.
    By the convexity of $\area_p$, $\vec{\point}_p$ is also in $\area_p$.
     
    Further, recall that when we restrict our domain to points $\vec{\point} \in \area_p$, we have that the revenue and consumer surplus are both linear function because the seller chooses a fix price $p$ within $\area_p$.
    Therefore, removing all segments in $\area_p$ and adding a new segment at $\vec{\point}_p$ with probability mass $\mu_x$ do not change the objective.
    Repeating this process for all areas proves the lemma.

\noindent\textbf{Proof of Theorem~\ref{thm:bayesian}}
    Consider the following more direct mathematical program that uses $\vec{\point}_p$'s and $\mu(\vec{\point}_p)$'s as the variables:
    \begin{align}
        \max & \sum_{p \in \valueset} \mu(\vec{\point}_p) \bigg( \objweight  \Rev\big(\valuedist(\vec{\point}_p)\big) +  (1-\objweight) \CS\big(\valuedist(\vec{\point}_p)\big) \bigg) \label{eq:LP-in-x} \\
\text{s.t.} \quad &\forall ~ p \in \valueset, \vec{\point}_p \in \area_p 
\text{ and }
\sum_{p \in \valueset} \mu(\vec{\point}_p) \vec{\point}_p = \centroid 
~. \nonumber
    \end{align}

    We need to show that (i) there is a one-to-one mapping between the variable space of program~\eqref{eq:LP-in-x} and that of LP~\eqref{eq:LP}, and (ii) under this mapping, the above program becomes LP~\eqref{eq:LP}.

    The mapping from $\vec{\point}_p$ and $\mu(\vec{\point}_p)$ to $\vec{\pointbelow}_p$ is already given, i.e., $\vec{\pointbelow}_p = \mu(\vec{\point}_p) \vec{\point}_p$.
    The other direction goes as follows.
    Given any $\vec{\pointbelow}_p$, let $\mu(\vec{x}_p) = \| \vec{\pointbelow}_p \|_1$ and $\vec{x}_p = \vec{z}_p / \| \vec{\pointbelow}_p \|_1$.
    Here, we use the fact that $\vec{x}_p$ lies on the probability simplex.

    Under this mapping, the objective of \eqref{eq:LP-in-x} becomes that of LP \eqref{eq:LP} due to linearity of $\Rev \big( \valuedist ( \cdot ) \big)$ and $\CS \big( \valuedist ( \cdot ) \big)$ for any fixed $p$.
    Finally, the equivalence of the constraints follow by the definition of the mapping.

\section{Missing Proofs Related to Robustification}
\label{app:robustify}

\subsection{Proof of Lemma~\ref{lem:robust-segmentation-performance-main}}
\label{app:sampleproof}


\begin{proof}
    We partition $\segmentset^*$ into $\segmentset^*_s$ and $\segmentset^*_i$, the former consists of the significant segments and the latter the insignificant ones, in the sense defined in Definition~\ref{def:robust-segmentation}.
    Similarly define $\segmentset^r_s$ and $\segmentset^r_i$.

    \medskip

    \emph{Welfare:~}
    The key lies in comparing the contributions from the significant segments, i.e., for any $\segment \in \segmentset^*_s$, we will show that:
    \begin{equation}
        \label{eqn:robust-segmentation-significant-segments-welfare}
        \SW \big( p^r_\segment, \valuedist(\vec{\point}^r_\segment) \big) \ge \SW \big( p^*_\segment, \valuedist(\vec{\point}^*_\segment) \big) - O\big(\epsilon_I\big)
        ~,
    \end{equation}

    It follows from the following inequalities:
    \begin{align*}
        &\quad\SW \big( p^r_\segment, \valuedist(\vec{\point}^r_\segment) \big)         
        = \int_{p^r_\segment}^1 v ~ d \valuedist(\vec{\point}^r_\segment) \\
        &
        \ge \int_{p^*_\segment}^1 v ~ d \valuedist(\vec{\point}^r_\segment) - \epsilon_I
        && \textrm{((3) of Definition~\ref{def:robust-segmentation})} \\
        &
        \ge \int_{p^*_\segment}^1 v ~ d \valuedist(\vec{\point}^*_\segment) - O\big(\epsilon_I\big)
        && \textrm{((2) of Definition~\ref{def:robust-segmentation})} \\[1ex]
        &
        = \SW \big( p^*_\segment, \valuedist(\vec{\point}^*_\segment) \big) - O\big(\epsilon_I\big)
        ~.
    \end{align*}

    Given Eqn.~\eqref{eqn:robust-segmentation-significant-segments-welfare}, the remaining the loss due to insignificant segments and the decrease in weights can be bounded as follows:
    \begin{align*}
        &\quad\sum_{\segment \in \segmentset^r} w^r_\segment \cdot \SW \big( p^r_\segment, \valuedist(\vec{\point}^r_\segment) \big)\\
        &\ge \sum_{\segment \in \segmentset^*_s} w^r_\segment \cdot \SW \big( p^r_\segment, \valuedist(\vec{\point}^r_\segment) \big) 
        && (\segmentset^*_s\subseteq \segmentset^* \subseteq \segmentset^r) \\
        & 
        \ge (1 - \epsilon_I) \sum_{\segment \in \segmentset^*_s} w^*_\segment \cdot \SW \big( p^r_\segment, \valuedist(\vec{\point}^r_\segment) \big) 
        && \textrm{((1) of Definition~\ref{def:robust-segmentation})} \\
        & 
        \ge \sum_{\segment \in \segmentset^*_s} w^*_\segment \cdot \SW \big( p^r_\segment, \valuedist(\vec{\point}^r_\segment) \big) - \epsilon_I
        && \textrm{(bounded values)} \\
        &
        \ge \sum_{\segment \in \segmentset^*_s} w^*_\segment \cdot \SW \big( p^*_\segment, \valuedist(\vec{\point}^*_\segment) \big) - O(\epsilon_I)
        && \textrm{(Eqn.~\eqref{eqn:robust-segmentation-significant-segments-welfare})} \\
        &
        \ge \sum_{\segment \in \segmentset^*} w^*_\segment \cdot \SW \big( p^*_\segment, \valuedist(\vec{\point}^*_\segment) \big) - O(\epsilon_I)
        ~.
        && \textrm{(insignificant segments)}
    \end{align*}

    \medskip

    \emph{Revenue:~}
    Similarly, the key is to compare the contribution from significant segments.
    For any $\segment \in \segmentset^*_s$, we have:
    \begin{align*}
        &\quad \Rev \big( p^r_\segment, \valuedist(\vec{\point}^r_\segment) \big)\\
        &\ge \Rev \big( p^*_\segment, \valuedist(\vec{\point}^r_\segment) \big) - \epsilon_S
        && \textrm{($\epsilon_S$-optimality of $p^r_\segment$)} \\
    & \ge \Rev \big( p^*_\segment, \valuedist(\vec{\point}^*_\segment) \big) - O \big(\epsilon_I\big)
        ~,
        && (\textrm{(2) of Definition~\ref{def:robust-segmentation}}; \epsilon_S\le \epsilon_I) 
    \end{align*}
    then follow almost a verbatim to the welfare counterpart, we can obtain:
    \[
        \sum_{\segment \in \segmentset^r} w^r_\segment \cdot \Rev \big( p^r_\segment, \valuedist(\vec{\point}^r_\segment) \big) 
        \ge
        \sum_{\segment \in \segmentset^*} w^*_\segment \cdot \Rev \big( p^*_\segment, \valuedist(\vec{\point}^*_\segment) \big) - O(\epsilon_I) ~.
    \]
    
    Next, we prove:
    \begin{align*}
        &\sum_{\segment \in \segmentset^*} w^*_\segment \cdot \Rev \big( p^*_\segment, \valuedist(\vec{\point}^*_\segment) \big) 
        \ge
        \sum_{\segment \in \segmentset^r} w^r_\segment \cdot \Rev \big( p^r_\segment, \valuedist(\vec{\point}^r_\segment) \big) - O(\epsilon_I) ~.
    \end{align*}
    
    Note that we also have:
    \begin{align*}
        &\quad \Rev \big( p^*_\segment, \valuedist(\vec{\point}^*_\segment) \big)\\
        &
        \ge \Rev \big( p^r_\segment, \valuedist(\vec{\point}^*_\segment) \big)
        && \textrm{(optimality of $p^*_\segment$)} \\
    & \ge \Rev \big( p^r_\segment, \valuedist(\vec{\point}^r_\segment) \big) - O \big(\epsilon_I\big)
        ~.
        && \textrm{((2) of Definition~\ref{def:robust-segmentation}})
    \end{align*}

    By our construction of $w^r_\segment$, we have $w^*_\segment\ge w^r_\segment$ for $\segment \in \segmentset^*$. Therefore,
    \begin{equation}
    \label{eqn:sampleproof-revenue-eq1}
    \begin{aligned}
        &\sum_{\segment \in \segmentset^*_s} w^*_\segment \cdot \Rev \big( p^*_\segment, \valuedist(\vec{\point}^*_\segment) \big) 
        \ge
         \sum_{\segment \in \segmentset^*_s} w^r_\segment \cdot \Rev \big( p^r_\segment, \valuedist(\vec{\point}^r_\segment) \big) - O(\epsilon_I) ~.
    \end{aligned}
    \end{equation}

    For insignificant segments, we have $\segment\in \segmentset^*_i$ and $\vec{\point}^*_\segment=\vec{\point}^r_\segment$. Then for all $p$, by the definition of insignificant segments:
    \begin{align*}
        \Rev \big( p , \valuedist(\vec{\point}^*_\segment) \big)=\Rev \big( p , \valuedist(\vec{\point}^r_\segment) \big)  
        &
        \le \E_{v \sim \valuedist(\vec{\point}^*_\segment)} [v] \\
        & \le \epsilon_I
        ~.
    \end{align*}
    
    Then, taking summation over $\segmentset^*_i$, we get that: 
    \begin{equation}
        \label{eqn:sampleproof-revenue-eq2}
    \begin{aligned}
        &\bigg|
        \sum_{\segment \in \segmentset^*_i} w^*_\segment \cdot \Rev \big( p^*_\segment, \valuedist(\vec{\point}^*_\segment) \big) 
        - 
         \sum_{\segment \in \segmentset^*_i} w^r_\segment \cdot \Rev \big( p^r_\segment, \valuedist(\vec{\point}^r_\segment) \big) 
        \bigg|
        = O(\epsilon_I) ~.
    \end{aligned}
    \end{equation}

    Combine Equations (\ref{eqn:sampleproof-revenue-eq1}) and (\ref{eqn:sampleproof-revenue-eq2}) we can obtain:
    \begin{align*}
        \sum_{\segment \in \segmentset^*} w^*_\segment \cdot \Rev \big( p^*_\segment, \valuedist(\vec{\point}^*_\segment) \big) 
        \ge
        \sum_{\segment \in \segmentset^*} w^r_\segment \cdot \Rev \big( p^r_\segment, \valuedist(\vec{\point}^r_\segment) \big)
        - O(\epsilon_I)~.
    \end{align*}

    Therefore it remains to show that
    \begin{equation}
    \label{eqn:sampleproof-revenue-eq3}
        \sum_{\segment \in \segmentset^r\setminus \segmentset^*}w^r_\segment \cdot \Rev\big( p^r_\segment, \valuedist(\vec{\point}^r_\segment) \big)  = O(\epsilon_I) ~.
    \end{equation} 
    
    In fact we have:
    \begin{align*}
        1   &= \sum_{\segment \in \segmentset^r}w^r_\segment = \sum_{\segment \in \segmentset^*}w^r_\segment+\sum_{\segment \in \segmentset^r\setminus \segmentset^*}w^r_\segment\\
            &\ge (1-\epsilon_I)\sum_{\segment \in \segmentset^*}w^*_\segment+\sum_{\segment \in \segmentset^r\setminus \segmentset^*}w^r_\segment
            \\
            &= (1-\epsilon_I) + \sum_{\segment \in \segmentset^r\setminus \segmentset^*}w^r_\segment ~.
    \end{align*}
    
    The second line is because $w^r_\segment\ge (1-\epsilon_I)w^*_\segment, \textrm{for }\segment\in \segmentset \subseteq \segmentset^*$.
    By boundedness of values, i.e., $v\le 1$, we get Eqn.~\eqref{eqn:sampleproof-revenue-eq3}.
\end{proof}

\subsection{Proof of Lemma~\ref{lem:robustify}}
\label{app:robustify-algorithm}
We formally define how the algorithm robustify an entire segmentation in Algorithm~\ref{alg:robustify-segmentation}, building on the subroutine that robustify a single (significant) segment introduced in Algorithm~\ref{alg:robustify-single-segment}.
The feasibility of the first step of Algorithm~\ref{alg:robustify-single-segment} is by Lemma~\ref{lem:robustify-find-type-main}, whose proof is deferred to Appendix~\ref{sec:robustify-find-type}.
\begin{algorithm*}[t]
    \caption{Robustify a segmentation}
    \label{alg:robustify-segmentation}
    \begin{flushleft}
        \begin{tabular}{ll}
            \textbf{Input:} & MHR-like distributions $\valuedist(\type)$, $\type \in \typeset$; \\
            & segmentation $(\segmentset^*, \segmentmap^*)$, represented by point-weight pairs $(\vec{\point}^*_\segment, w^*_\segment)$, $\segment \in \segmentset^*$; \\
            & prices $p^*_\segment$, $\segment \in \segmentset^*$, which are optimal w.r.t.\ the corresponding $\valuedist(\vec{\point}^*_\segment)$'s. \\[1ex]
            \textbf{Output:} & Robustified segmentation $(\segmentset^r, \segmentmap^r)$, represented by $\vec{\point}^r_\segment$'s and weights $w^r_\segment$'s.
        \end{tabular}
    \end{flushleft}
    \begin{enumerate}
        \item For every segment $\segment \in \segmentset^*$:
        \begin{enumerate}
            \item If segment $\segment$ is insignificant, i.e., $\E_{v \sim \valuedist(\vec{\point}^*_\segment)} [v] < \epsilon_I$, let $\vec{\point}^r_\segment = \vec{\point}^*_\segment$.
            \item Otherwise, construct $\vec{x}^r_\segment$ using Algorithm~\ref{alg:robustify-single-segment};
            \item Decrease the weight by a $1 - \epsilon_I$ multiplicative factor, i.e., $w^r_\segment = (1-\epsilon_I) w^*_\segment$ in both cases.
        \end{enumerate}
        \item For every type $\type \in \typeset$:
        \begin{enumerate} 
            \item Add a new segment $\segment(\type) \notin \segmentset^*$ to $\segmentset^r$ such that $\vec{\point}^r_{\segment(\type)}$ is the vertex of simplex $\simplex$ that corresponds to type $\type$.
            \item Let its weight be $w^r_{\segment(\type)} = \frac{1}{\typenum} - \sum_{\segment \in \segmentset^*} w^r_\segment \cdot x^r_{\segment, \type}$.
        \end{enumerate}
    \end{enumerate}
\end{algorithm*}

\begin{algorithm*}[t]
    \caption{Robustify a (significant) segment}
    \label{alg:robustify-single-segment}
    \begin{flushleft}
        \begin{tabular}{ll}
            \textbf{Input:} & MHR-like distributions $\valuedist(\type)$, $\type \in \typeset$; \\
            & segment $\segment$, represented by $\vec{\point}^*_\segment$, such that $\E_{v \sim \valuedist(\vec{\point}^*_\segment)} [ v ] \ge \epsilon_I$; \\
            & price $p^*$, which is optimal w.r.t.\ $\valuedist(\vec{\point}^*_\segment)$. \\[1ex]
            \textbf{Output:} & Robustified version of $\segment$, represented by $\vec{\point}^r_\segment$.
        \end{tabular}
    \end{flushleft}
    \begin{enumerate}
        \item Find type $\type_\segment$ such that for any price $p$ whose quantile w.r.t.\ $\valuedist(\vec{\point}^*_\segment)$ is at most that of $p^*$ less $\epsilon_I$, we have $\Rev \big( p, \valuedist(\type_\segment) \big) < \Rev \big( p^*, \valuedist(\type_\segment) \big) - \frac{\epsilon_S\typenum}{\epsilon_I}$.
        \item Construct the robustified version as $\vec{\point}^r_\segment = \big(1-\tfrac{\epsilon_I}{\typenum}\big) \cdot \vec{\point}^*_\segment + \tfrac{\epsilon_I}{\typenum} \cdot \vec{\point}_{t_\segment}$, where $\vec{\point}_{t_\segment} \in \simplex$ has its $\type_\segment$-th entry being one and all others being zeros.
    \end{enumerate}
\end{algorithm*}


Before proving Lemma~\ref{lem:robustify} we show the following.

\begin{lemma}
    \label{lem:robustify-single-segment}
    The robustified version of $\segment$ returned by Algorithm~\ref{alg:robustify-single-segment}, represented by $\vec{\point}^r_\segment$, satisfies the last two properties stated in Definition~\ref{def:robust-segmentation}, i.e.: 
    \begin{enumerate}
        \setcounter{enumi}{1}
        \item $\| \vec{\point}^*_\sigma - \vec{\point}^r_\sigma \|_1 \le \epsilon_I$.
        \item All $\epsilon_S$-optimal prices w.r.t.\ $\valuedist(\vec{\point}^r_\segment)$ have quantiles at least that of $p^*$ less $\epsilon_I$.
    \end{enumerate}
\end{lemma}

\begin{proof}
    Property (2) follows directly by the definition that:
    \[
        \vec{\point}^r_\segment = \big(1-\tfrac{\epsilon_I}{\typenum}\big) \cdot \vec{\point}^*_\segment + \tfrac{\epsilon_I}{\typenum} \cdot \vec{\point}_{t_\segment}
        ~.
    \]

    It remains to verify property (3). 
    We need to show that for any price $p$ whose quantile w.r.t.\ $\valuedist(\vec{\point}^r_\segment)$ is less than that of $p^*$ minus $\epsilon_I$, it cannot be $\epsilon_S$-optimal w.r.t.\ $\valuedist(\vec{\point}^r_\segment)$.
    
    By the definition of $\vec{\point}^r_\segment$, we have:
    \begin{align*}
        \Rev \big( p, \valuedist(\vec{\point}^r_\segment) \big)
        &= \left(1-\frac{\epsilon_I}{\typenum}\right) \cdot \Rev \big( p, \valuedist(\vec{\point}^*_\segment) \big) 
        + \frac{\epsilon_I}{\typenum} \cdot \Rev \big( p, \valuedist(\type_\segment) \big)
        ~.
    \end{align*}

    We next bound the two terms on the RHS separately.
    By that $p^*$ is optimal w.r.t.\ $\valuedist(\vec{\point}^*_\segment)$, the first term is upper bounded by:
    \[
        \Rev \big( p, \valuedist(\vec{\point}^*_\segment) \big) \le \Rev \big( p^*, \valuedist(\vec{\point}^*_\segment) \big)
        ~.
    \]

    Further, note that $p$'s quantile w.r.t.\ $\valuedist(\vec{\point}^*_\segment)$ is less than that of $p^*$ minus $\epsilon_I$, due to the assumed quantile gap w.r.t.\ $\valuedist(\vec{\point}^r_\segment)$, and that $\vec{\point}^*_\segment$ and $\vec{\point}^r_\segment$ is close (property (2)).
    Hence, by the definition of the algorithm (and Lemma~\ref{lem:robustify-find-type-main}), the second term is upper bounded by:
    \[
        \Rev \big( p, \valuedist(\type_\segment) \big) < \Rev \big( p^*, \valuedist(\type_\segment) \big) - \frac{\epsilon_S \typenum}{\epsilon_I}
        ~.
    \]
    %
    
    Putting together gives the followings:
    \begin{align*}
        \Rev \big( p, \valuedist(\vec{\point}^r_\segment) \big)
        &
        < \left(1-\frac{\epsilon_I}{\typenum}\right) \cdot \Rev \big( p^*, \valuedist(\vec{\point}^*_\segment) \big) 
        + \frac{\epsilon_I}{\typenum} \cdot \left( \Rev \big( p^*, \valuedist(\type_\segment) \big) - \frac{\epsilon_S \typenum}{\epsilon_I} \right) \\
        & 
        < \left(1-\frac{\epsilon_I}{\typenum}\right) \cdot \Rev \big( p^*, \valuedist(\vec{\point}^*_\segment) \big) 
        + \frac{\epsilon_I}{\typenum} \cdot \Rev \big( p^*, \valuedist(\type_\segment) \big) - \epsilon_S \\
        &
        = \Rev \big( p^*, \valuedist(\vec{\point}^r_\segment) \big) - \epsilon_S
        ~.
    \end{align*}
    
    Hence, such a price $p$ cannot be $\epsilon_S$-optimal w.r.t.\ $\valuedist(\vec{\point}^r_\segment)$.
\end{proof}


\begin{proof}[Proof of Lemma~\ref{lem:robustify}]
    We prove it using Algorithm~\ref{alg:robustify-segmentation}.
    Property (1) follows by the definition of the algorithm, in particular, steps (1c).
    Properties (2) and (3) follows by how the algorithm handles significant segments, i.e., step (1b), and Lemma~\ref{lem:robustify-single-segment}.

    It remains to verify that the algorithm gives a feasible segmentation. 
    First, we need to show that the weights, in particular, those of the new segments added in step (2), are non-negative.
    Second, we need to prove that centroid of the weighted point set lies in the middle of the simplex.

    Note that the second part follows from step (2b) of the algorithm, provided that the weights are non-negative. It remains to argue that:
    \[
        w^r_{\segment(\type)} = \frac{1}{\typenum} - \sum_{\segment \in \segmentset^*} w^r_\segment \cdot \point^r_{\segment, \type} \ge 0
        ~.
    \]

    It follows from a sequence of inequalities given below:
    \begin{align*}
        &\quad\sum_{\segment \in \segmentset^*} w^r_\segment \cdot \point^r_{\segment, \type}
        = \big( 1 - \epsilon_I \big) \cdot \sum_{\segment \in \segmentset^*} w^*_\segment \cdot \point^r_{\segment, \type} \\
        &
        \le \big( 1 - \epsilon_I \big) \cdot \sum_{\segment \in \segmentset^*} w^*_\segment \cdot \left( \point^*_{\segment, \type} + \frac{\epsilon_I}{T} \right)
        && \textrm{(Definition of $\vec{\point}^r_\segment$.)} \\
        &
        = \big( 1 - \epsilon_I \big) \cdot \frac{1 + \epsilon_I}{T}
        && \textrm{($\textstyle\sum w^*_\segment \point^*_{\segment, \type} = \frac{1}{T}, \sum w^*_\segment = 1$)} \\
        & 
        < \frac{1}{T}
        ~.
    \end{align*}
\end{proof}

\subsection{Proof of Lemma~\ref{lem:robustify-find-type-main}}
\label{sec:robustify-find-type}

Note that the lemma holds trivially if there is no price whose quantile w.r.t.\ $\valuedist(\vec{\point}^*_\segment)$ is at most that of $p^*$ less $\epsilon_I$.
The rest of the proof assumes this is not the case.
Also, for simplicity, in the proof we introduce a frequently used notation $\epsilon_R = \frac{\epsilon_S \typenum}{\epsilon_I}$.

Then, let $\bar{\price}$ denote the smallest price whose quantile is at most that of $p^*$ less $\epsilon_I$.
%
%
%
By the concavity of the revenue curve (Condition~\ref{con:concavity} of Definition~\ref{def:mhr-like-distributions}), it suffices to show that for price $\bar{\price}$, there exists a type $\type \in \typeset$ such that: 
\[
    \Rev \big( \bar{\price}, \valuedist(\type) \big) < \Rev \big( p^*, \valuedist(\type) \big) - \epsilon_R
    ~.
\]
Then, the inequality holds for all prices $p \ge \bar{\price}$ and, thus, the lemma follows.

The rest of the subsection assumes for contrary that Lemma~\ref{lem:robustify-find-type-main} does not hold. 
That is, we assume that for every type $\type \in \typeset$:
\[
    \Rev \big( \bar{\price}, \valuedist(\type) \big) \ge \Rev \big( p^*, \valuedist(\type) \big) - \epsilon_R
    ~.
\]

Let $p_1 = p^* + \tfrac{1}{\valuenum}, p_2 = p^* + \tfrac{2}{\valuenum}$, and so or, be the prices between $p^*$ (exclusive) and $\bar{\price}$ (inclusive).
Let $P$ denote the set of these prices.
By the concavity of the revenue curve (Condition~\ref{con:concavity} of Definition~\ref{def:mhr-like-distributions}), the above inequality implies that for any type $\type \in \typeset$, and any price $p_i \in P$:
\begin{equation}
    \label{eqn:robustify-contrary}
    \Rev \big( p_i, \valuedist(\type) \big) \ge \Rev \big( p^*, \valuedist(\type) \big) - \epsilon_R
    ~.
\end{equation}

Under this assumption we will show that there must be a type $\type$ such that $\valuedist(\type)$ satisfies the `large plateau' conditions(as mentioned in Section \ref{sec:proof-robustify-find-type-main}):
\begin{enumerate}
	\item Revenue between $\bar{p}$ and $p^*$ is not much higher: $\forall~ p \in [p^*,\bar{p}]$ , 
	\[ \textstyle \Rev \big( p, \valuedist(\type)\big) - \Rev \big( p^*, \valuedist(\type)\big) \le \tilde{O} \big( \frac{\epsilon_R \typenum}{\epsilon_I} \big)
	.\]
	%
	\item The plateau is high, i.e., revenue of $p^*$ is large:  $\Rev \big( p^*, \valuedist(\type)\big) \ge \Omega \big( \frac{\epsilon_I}{T} \big)$. 
	\item The plateau is wide, i.e., the quantiles of $\bar{p}$ and $p^*$ differ by at least $\Omega \big( \frac{\epsilon_I^2}{T} \big)$.
\end{enumerate}
Finally we argue that these conditions will lead to a contradiction.

We will first prove the following technical lemmas.


\begin{lemma}
    \label{lem:robustify-one-to-many}
    Suppose type $\type \in \typeset$ and price $p_i \in P$ satisfy that:
    \[
        \Rev \big( p_i, \valuedist(\type) \big) - \Rev \big( p^*, \valuedist(\type) \big) > 0
        ~.
    \]
    Then, for any $p_j \in P$, $j \le i$, we have:
    \begin{equation}
        \label{eqn:robustify-pairwise-relation}
        \begin{aligned}
        &\frac{1}{j} \cdot \left( \Rev \big( p_j, \valuedist(\type) \big) - \Rev \big( p^*, \valuedist(\type) \big) \right) \\
        &\quad\ge \frac{1}{i} \cdot \left( \Rev \big( p_i, \valuedist(\type) \big) - \Rev \big( p^*, \valuedist(\type) \big) \right)
        ~.
        \end{aligned}
    \end{equation}
\end{lemma}

\begin{proof}
    For simplicity of notations, let $R_i$, $R_j$, and $R^*$ denote $\Rev \big( p_i, \valuedist(\type) \big)$, $\Rev \big( p_j, \valuedist(\type) \big)$, and $\Rev \big( p^*,\valuedist(\type) \big)$, respectively. 
    Then, we can rewrite the inequality as:
    \[
        \frac{1}{j} \big( R_j - R^* \big) \ge \frac{1}{i} \big( R_i - R^* \big)
        ~.
    \]

    Further, Let $q_i = \tfrac{R_i}{p_i}$, $q_j = \tfrac{R_j}{p_j}$, and $q^* = \tfrac{R^*}{p^*}$ be the corresponding quantiles.

    By concavity of the revenue-quantile curve of MHR-like distributions, we have:
    \[
        p_j q_j \le \frac{q_j-q_i}{q^*-q_i} p^* q^* + \frac{q^*-q_j}{q^*-q_i} p_i q_i 
        ~.
    \]

    Multiplying both sides by $q^*-q_i$, and rearranging terms, we obtain:
    \[
        q_j \big( q^* (p_j - p^*) + q_i (p_i - p_j) \big) \ge q_i q^* (p_i - p^*) 
        ~.
    \]

    By the definition of $p_i$ and $p_j$, it becomes:
    \begin{equation}
        \label{eqn:robustify-straight-line}
        q_j \ge \frac{i q_i q^*}{j q^* + (i-j) q_i}
        ~.
    \end{equation}

    Suppose we fix $i$, $j$, and $q_i$, $q^*$. 
    Then, the larger $q_j$ is, the larger the LHS of \eqref{eqn:robustify-pairwise-relation} is, while the RHS stays the same.
    Hence, it suffices to prove the lemma when the above Eqn.~\eqref{eqn:robustify-straight-line} holds with equality. 

    For such a $q_j$, the following three points $(q_i,R_i)$, $(q_j,R_j)$ and $(q^*,R^*)$ are on the same line, since the inequality is derived from the concavity of the revenue curve.
    In other words, we have:
    \[
        \frac{R_i - R^*}{R_j - R^*} = \frac{q^* - q_i}{q^* - q_j}
        ~.
    \]

    Therefore, it remains to show the following inequality when Eqn.~\eqref{eqn:robustify-straight-line} holds with equality:
    \[
        \frac{1}{j} \big( q^* - q_j \big) \ge \frac{1}{i} (q^*-q_i)
        ~,
    \]
    which is equivalent to:
    \[
        q_j \le \frac{(i-j) q^* + j q_i}{i} ~.
    \]

    Putting together, it suffices to show:
    \[
        \frac{i q_i q^*}{j q^* + (i-j) q_i} \le \frac{(i-j) q^* +j q_i}{i}
        ~,
    \]
    which is equivalent to:
    \[
        i^2 q_i q^* \le \big( j q^* + (i-j) q_i \big) \big( (i-j) q^* +j q_i \big) 
        ~.
    \]

    Rearranging terms, this is simply:
    \[
        (i-j)j \big(q^*-q_i\big)^2 \ge 0
        ~.
    \]

    So the lemma follows.
\end{proof}

The rest of the proof considers sampling a type $\type$ according to $\vec{\point}^*_\segment$, and establishes a sequence of technical claims regarding the probabilities of various events.
The main technical ingredient is the next lemma, which states that the chance of having a price in $P$ whose revenue is much higher than that of $p^*$ is small.

\begin{lemma}
    \label{lem:robustify-beat-union-bound}
    For any $\delta > 0$, $\frac{\epsilon_R \log \valuenum}{\delta}$ upper bounds:
    \[
        \Pr_{\type \sim \vec{\point}^*_\segment} \left[ \exists p_i \in P: \Rev \big( p_i, \valuedist(\type) \big) \ge \Rev \big( p^*, \valuedist(\type) \big) + \delta \right]
        ~.
    \]
    Choosing $\delta=\tilde{\Theta}\left(\frac{\epsilon_R\typenum}{\epsilon_I}\right)$ we obtain the probability that
condition $1$ is violated is less than $O(\frac{\epsilon_I}{\typenum})$.
\end{lemma}

\begin{proof}
    For every $p_i \in P$, let $\alpha_i=\frac{1}{i\log \valuenum}$.  
    Then, multiplying both sides of Eqn.~\eqref{eqn:robustify-contrary} by $\alpha_i$ and summing over all $p_i \in P$, we have that for any type $\type \in \typeset$:
    \begin{align*}
     &\sum_{p_i \in P} \alpha_i \cdot \big( \Rev \big(p_i, \valuedist(\type) \big)-\Rev \big(p^*,\valuedist(\type)\big)\big) \\
    &\qquad  \ge - \epsilon_R \sum_{p_i \in P} \alpha_i \ge -\epsilon_R ~.
    \end{align*}

    To simplify notations in the rest of the proof, let $\Delta_t$ denote the LHS of the above inequality for a given type $\type$.
    Then, on one hand, we have:
    \[
        \E_{\type \sim \vec{\point}^*_\sigma} \big[ \Delta_t \big] 
        = \sum_{p_i \in P} \alpha_i \cdot \big( \Rev \big(p_i, \valuedist(\vec{\point}^*_\segment) \big)-\Rev \big(p^*,\valuedist(\vec{\point}^*_\segment)\big)\big)
        ~.
    \]
    It is nonpositive because of the optimality of $p^*$ w.r.t.\ $\valuedist(\vec{\point}^*_\segment)$.

    On the other hand, let $\delta' = \frac{\delta}{\log \valuenum}$, we have:
    \begin{align*}
        \E_{\type \sim \vec{\point}^*_\segment} \big[ \Delta_\type \big]
        &
        \ge \Pr_{\type \sim \vec{\point}^*_\segment} \big[ \Delta_\type \ge \delta' \big] \cdot \delta' + \min_\type \Delta_\type \\
        &
        \ge \Pr_{\type \sim \vec{\point}^*_\segment} \big[ \Delta_\type \ge \delta' \big] \cdot \delta' + (-\epsilon_R) 
        ~,
    \end{align*}
    by the definition of $\Delta_t$ and Eqn.~\eqref{eqn:robustify-contrary}.

    Putting together we get that:
    \begin{equation}
        \label{eqn:robustify-beat-union-bound-weighted-sum}
        \Pr_{\type \sim \vec{\point}^*_\segment} \big[ \Delta_\type \ge \delta' \big] \le \frac{\epsilon_R}{\delta'} = \frac{\epsilon_R \log \valuenum}{\delta}
        ~.
    \end{equation}

    The rest of the proof boils down to showing that whenever there exists $p_i \in P$ such that:
    \[
        \Rev \big( p_i, \valuedist(\type) \big) \ge \Rev \big(p^*, \valuedist(\type) \big) + \delta 
        ~, 
    \]
    it implies that $\Delta_t \ge \delta'$.
    This, together with Eqn.~\eqref{eqn:robustify-beat-union-bound-weighted-sum} would complete the proof of the lemma.




    Indeed, when such a $p_i$ exists, we have:
    \begin{align*}
        \Delta_t 
        & 
        \ge \sum_{j=1}^i \alpha_j \cdot \big( \Rev \big(p_j,\valuedist(\type) \big) - \Rev \big(p^*, \valuedist(\type) \big) \big) \\
        & 
        \ge \sum_{j=1}^{i} \alpha_j \cdot\frac{j}{i} \cdot \big( \Rev \big( p_i, \valuedist(\type) \big) - \Rev \big(p^*, \valuedist(\type) \big) \big) \\
        & 
        = \frac{1}{\log \valuenum} \cdot \big( \Rev \big( p_i, \valuedist(\type) \big) - \Rev \big( p^*, \valuedist(\type) \big) \big) \\
        & 
        \ge \frac{\delta}{\log \valuenum} 
        ~.
    \end{align*}

    The second line is by Lemma~\ref{lem:robustify-one-to-many} and the third is by the definition of $\alpha_j$'s. So the lemma follows.
%
\end{proof}

\paragraph{Some Notations}
In the following discussions, let $p(\type)$ denote the optimal price w.r.t.\ $\valuedist(\type)$, for any type $\type \in \typeset$. 
Further, let $R(\type)$ denote the corresponding optimal revenue, and similarly $R(\vec{\point})$ for mixtures.
Finally, let $q^*(\type)$ and $\bar{q}(\type)$ denote the quantiles of $p^*$ and $\bar{\price}$ w.r.t. $\valuedist(\type)$;
let $q^*$ and $\bar{q}$, i.e., without specifying a type, denote the quantiles w.r.t.\ $\valuedist(\vec{\point}^*_\segment)$.

\begin{lemma}
    \label{lem:robustify-optimal-SW-lower-bound}
    For any $\vec{\point} \in \simplex$, the optimal revenue of the corresponding mixture $\valuedist(\vec{\point})$ is lower bounded as follows:
    \[
        R(\vec{\point}) \ge \Omega \left( \frac{\E_{v \sim \valuedist(\vec{\point})}[v]}{\typenum} \right)
        ~.
    \]
\end{lemma}

\begin{proof}
    First, we can rewrite the social welfare of distribution $\valuedist(\vec{\point})$ as follows:
    \[
        \E_{v \sim \valuedist(\vec{\point})}[v] = \sum_{\type \in \typeset} \point_\type \cdot \E_{v \sim \valuedist(\type)}[v]
        ~.
    \]

    Next, note that for any given $\type' \in \typeset$:
    \begin{align*}
        \Rev \big( p(\type'), \valuedist(\vec{\point}) \big) 
        & 
        = p(\type') \cdot \Pr_{v \sim \valuedist(\vec{\point})} \big[ v \ge p(\type') \big] \\
        & \textstyle
        = p(\type') \cdot \sum_{\type \in \typeset} \point_\type \cdot \Pr_{v \sim \valuedist(\type)} \big[ v \ge p(\type') \big] \\
        & 
        \ge p(\type') \cdot \point_{\type'} \cdot \Pr_{v \sim \valuedist(\type')} \big[ v \ge p(\type') \big] \\
        &
        = \point_{\type'} \cdot R(\type') \\
        &
        \ge \point_{\type'} \cdot \Omega \big( \E_{v \sim \valuedist(\type')} [v] \big) ~,
    \end{align*}
    where the last line is by Condition~\ref{con:revenue-welfare-ratio} of Definition~\ref{def:mhr-like-distributions}.
    Therefore, we have:
    \[
        T \cdot R(\vec{\point}) \ge \sum_{\type \in \typeset} \Rev \big( p(t), \valuedist(\vec{\point}) \big) \ge \Omega \big( \E_{v \sim \valuedist(\vec{\point})}[v] \big) 
        ~.
    \]

    Dividing both sides by $T$ proves the lemma.
\end{proof}

As a direct corollary of the above lemma, the assumption that $\E_{v \sim \valuedist(\vec{\point}^*_\segment)} [v] \ge \epsilon_I$, and that $p^*$ is optimal w.r.t.\ $\valuedist(\vec{\point}^*_\segment)$, we have the following lemma.

\begin{lemma}
    \label{lem:robustify-optimal-revenue-lower-bound} 
    The optimal revenue of distribution $\valuedist(\vec{\point}^*_\segment)$ is lower bounded as:
    \[
        R \big( \vec{\point}^*_\segment \big)=
        \Rev \big( p^*, \valuedist(\vec{\point}^*_\segment) \big) 
        \ge \Omega \left( \frac{\epsilon_I}{\typenum} \right)
        ~.
    \]
\end{lemma}

\begin{lemma}
    \label{lem:robustify-condition-two}
    There's at least an $\Omega(\frac{\epsilon_I}{\typenum})$ probability that:
    \[
        \Rev(p^*,\valuedist(\type)) \ge \Omega(\frac{\epsilon_I}{\typenum})~.
    \]
\end{lemma}
\begin{proof}
In fact by Lemma \ref{lem:robustify-optimal-revenue-lower-bound} we have:
\[
    \Rev \big( p^*, \valuedist(\vec{\point}^*_\segment) \big)=
    \E_{\type\sim\vec{\point}^*_\segment}[\Rev \big( p^*, \valuedist(\type) \big)]\ge
    \Omega \left( \frac{\epsilon_I}{\typenum} \right) ~.
\]
This will be violated if $\Rev(p^*,\valuedist(\type)) \ge \Omega(\frac{\epsilon_I}{\typenum})$ does not have at least an $\Omega(\frac{\epsilon_I}{\typenum})$ probability.
\end{proof}


Next, we show a conditional lower bound on the quantile gap between price $\bar{\price}$ and $p^*$.
\begin{lemma}
    \label{lem:robustify-quantile-gap}
    Suppose type $\type \in \typeset$ satisfies condition $1$ and $2$, i.e.:
    \[
        \forall p\in [p^*,\bar{p}], 
        \Rev \big(\price, \valuedist(\type) \big) - \Rev \big( p^*, \valuedist(\type) \big)  
        \le
        \tilde{O} \bigg( \frac{\epsilon_R\typenum}{\epsilon_I} \bigg)
        ~,
    \]
    and
    \[
        \Rev\big( p^*, \valuedist(\type)\big) 
        \ge 
        \Omega(\frac{\epsilon_I}{\typenum})
    \]
    Then this type satisfies condition $3$:
    \[
        q^*(t) - \bar{q}(t) \ge \Omega \big( \frac{\epsilon_I^2}{\typenum}\big)
        ~.
    \]
\end{lemma}



\begin{proof}
    We first prove:
    \[
        \frac{\bar{q}(t)}{q^*(t)} \le 1 - \Omega \big( \epsilon_I \big)
        ~.
    \]

    By the condition in the lemma, we have:
    \[
        \bar{\price} \bar{q}(t) \le p^* q^*(t) + O \bigg(  \frac{\epsilon_R\typenum}{\epsilon_I} \bigg) \le \left( 1 + \frac{\epsilon_R\typenum^2}{\epsilon_I^2} \right) \cdot p^* q^*(t)
        ~,
    \]
    where the second inequality follows from Lemma~\ref{lem:robustify-condition-two} (noting that $p^* q^*(t) = \Rev \big( p^*, \valuedist(\type) \big)$).

    Hence, we have:
    \begin{align*}
        \frac{\bar{q}(t)}{q^*(t)} & \le \left( 1 + O\bigg(\frac{\epsilon_R\typenum^2}{\epsilon_I^2}\bigg) \right) \cdot  \frac{p^*}{\bar{\price}} \\
        & = \left( 1 + O\bigg(\frac{\epsilon_R\typenum^2}{\epsilon_I^2}\bigg) \right) \cdot \frac{\Rev \big( p^*, \valuedist(\vec{\point}^*_\segment) \big)}{\Rev \big( \bar{\price}, \valuedist(\vec{\point}^*_\segment) \big)} \cdot \frac{\bar{q}}{q^*}
        ~.
    \end{align*}

    The second term on the RHS is further bounded as follows:
    \begin{align*}
        \frac{\Rev \big( p^*, \valuedist(\vec{\point}^*_\segment) \big)}{\Rev \big( \bar{\price}, \valuedist(\vec{\point}^*_\segment) \big)}
        & 
        \le \frac{\Rev \big( p^*, \valuedist(\vec{\point}^*_\segment) \big)}{\Rev \big( p^*, \valuedist(\vec{\point}^*_\segment) \big) - \epsilon_R} \quad \textrm{(Eqn.~\eqref{eqn:robustify-contrary})} \\
        & 
        \le \left( 1 + O \left(\frac{\epsilon_R \typenum}{\epsilon_I} \right) \right)~.
        \quad \textrm{(Lemma~\ref{lem:robustify-optimal-revenue-lower-bound})}
    \end{align*}

    The third term is bounded as:
    \begin{align*}
        \frac{\bar{q}}{q^*}
        &
        \le \frac{q^* - \epsilon_I}{q^*} 
        && \textrm{(Definition $\bar{\price}$)}  \\
        &
        \le 1 - \epsilon_I
        ~.
        && \textrm{($q^* \le 1$)}
    \end{align*}
    
    By Eqn.~\eqref{eqn:robustify-epsilons}:
    \[
        \epsilon_I \ge O \left(\frac{\epsilon_R \typenum}{\epsilon_I} \right) , 
        \epsilon_I \ge O \left(\frac{\epsilon_R\typenum^2}{\epsilon_I^2} \right)
        ~,
    \]
    so we have $\frac{\bar{q}(t)}{q^*(t)} \le 1 - \Omega \big( \epsilon_I \big)$. 
    Note that $p^* q^*(\type) \ge \Omega(\frac{\epsilon_I}{\typenum})$ and $p^*\le 1$, we have $q^*\ge \Omega(\frac{\epsilon_I}{\typenum})$, then
    \[
         q^*(t) - \bar{q}(t) 
         \ge \Omega \big( \epsilon_I \big)q^*(\type) 
         \ge \Omega \big( \frac{\epsilon_I^2}{\typenum}\big)
        ~.
    \]
\end{proof}




Up to now we have proved the existence of a type $\type$ that satisfies all the $3$ `large plateau' conditions.
Let $\type$ be such a type.
The next lemma handles the cases when $p(t)$ is large, medium, or small respectively, depending on its relation with $\bar{\price}$ and $p^*$.
We show that the revenue gap between prices $p^*$ and $\bar{p}$ is at least $\epsilon_R$,
and will lead to contradiction.

\begin{lemma}
     \label{lem:robustify-contradiction}
     For type $\type$ that satisfies all the $3$ conditions, we will have:
     \[
        \Rev \big( p^*, \valuedist(\type) \big)-\Rev \big( \bar{p}, \valuedist(\type) \big) 
        =
        \Omega\left( \frac{\epsilon_S\typenum}{\epsilon_I}\right) ~.
     \]
\end{lemma}

\begin{proof}
Consider the relationship between $p(\type)$ and $[p^*,\bar{p}]$.

\paragraph{(1)}$p(\type)\ge \bar{p}$. 
In this case $q(\type)\le \bar{q}(\type)\le q^*(\type)$. 
The quantile of $p(t)$ w.r.t.\ $\valuedist(\type)$ is at least $\Omega(1)$ (Condition \ref{con:monopoly-sale-prob} of Defintion~\ref{def:mhr-like-distributions}), we have that $q^*(\type) = \Omega(1)$.
Recall in Lemma~\ref{lem:robustify-quantile-gap} we obtain $q^*(t) - \bar{q}(t) \ge \Omega \big( \epsilon_I \big)q^*(\type)$. 
So $q^*(t) - \bar{q}(t) \ge \Omega \big( \epsilon_I \big)$.
Further, by $\Rev \big( p^*, \valuedist(\type) \big) \ge \Omega (\frac{\epsilon_I}{\typenum})$, we have:
    \[
        R(\type) \ge \Omega (\frac{\epsilon_I}{\typenum})
        ~.
    \]

    By that both $\bar{\price}$ and $p^*$ are no more than the monopoly price $p(\type)$, the concavity of the revenue curve, and the strong concavity near monopoly price (Conditions \ref{con:concavity} and \ref{con:strong-concavity} of Definition \ref{def:mhr-like-distributions}), we get that the revenue of $\bar{\price}$ w.r.t.\ $\valuedist(\type)$ is larger than that of $p^*$ by at least:
    \[
        \Omega \left( \frac{\epsilon_I}{\typenum} \cdot \epsilon_I^2 \right) = \Omega \left( \frac{\epsilon_I^3}{\typenum} \right) 
        \ge \epsilon_R
        ~,
    \]
    where the inequality follows by Eqn.~\eqref{eqn:robustify-epsilons}.

\paragraph{(2)}$p(\type)\in [p^*,\bar{p})$. 
Follow the same statements as (1) we have:
    \[
        q^*(t) - \bar{q}(t) \ge \Omega \big( \epsilon_I \big),\quad 
        R(\type) \ge \Omega (\frac{\epsilon_I}{\typenum})
        ~.
    \]
Note that $q^*(\type)\ge q(\type)> \bar{q}(\type)$, the gap between $q(\type)$ and either $q^*(t)$ or $\bar{q}(\type)$ (or both) is at least $\Omega(\epsilon_S)$. 
On one hand, suppose it is the former.
The strong concavity of the revenue curve near the monopoly price (Condition \ref{con:strong-concavity} of Definition \ref{def:mhr-like-distributions}) implies that:
    \begin{align*}
        \Rev \big( p(\type), \valuedist(\type) \big) - \Rev \big( p^*, \valuedist(\type) \big) 
        & 
        \ge \Omega \left( \frac{\epsilon_I}{\typenum} \cdot \epsilon_I^2 \right) \\
        & = \Omega \left( \frac{\epsilon_I^3}{\typenum} \right) 
        ~.
    \end{align*}

On the other hand, suppose it is the latter:
The strong concavity of the revenue curve near the monopoly price (Condition \ref{con:strong-concavity} of Definition \ref{def:mhr-like-distributions}) implies that:
    \begin{align*}
        \Rev \big( p(\type), \valuedist(\type) \big) - \Rev \big( \bar{\price}, \valuedist(\type) \big) 
        & 
        \ge \Omega \left( \frac{\epsilon_I}{\typenum} \cdot \epsilon_I^2 \right) \\
        & = \Omega \left( \frac{\epsilon_I^3}{\typenum} \right)
        ~.
    \end{align*}
 Both cases violate the condition $1$ of the `large plateau' notion, because $ \Omega \left( \frac{\epsilon_I^3}{\typenum} \right)\ge \tilde{O}(\frac{\epsilon_R\typenum}{\epsilon_I})$ by Eqn.~\eqref{eqn:robustify-epsilons}. 
 
\paragraph{(3)}$p(\type)< p^*$. In this case we still have:
    \[
        q^*(t) - \bar{q}(t) \ge \Omega \big( \epsilon_I \big)q^*(\type),\quad 
        R(\type) \ge \Omega (\frac{\epsilon_I}{\typenum})
        ~.
    \]
And we can only use the quantile gap $ q^*(t) - \bar{q}(t) \ge \Omega(\frac{\epsilon_I^2}{\typenum})$ since $q^*(\type)$ is not always $\Omega(1)$.
By condition \ref{con:strong-concavity} of Definition \ref{def:mhr-like-distributions}, we get that the revenue of $p^*$ w.r.t.\ $\valuedist(\type)$ is larger than that of $\bar{\price}$ by at least:
    \[
        \Omega \left( \frac{\epsilon_I}{\typenum} \cdot \left( \frac{\epsilon_I^2}{\typenum} \right)^2 \right) = \Omega \left( \frac{\epsilon_I^5}{\typenum^3} \right) > \epsilon_R 
        ~,
    \]
    where the inequality follows from Eqn.~\eqref{eqn:robustify-epsilons}.
\end{proof}

\section{Poly-time Approximate Projection to MHR-like Distributions}
\label{app:find-nearby-mhr}
\begin{algorithm*}[t]
    \caption{Find a Nearby MHR-like Distribution}
    \label{alg:find-nearby-mhr}
    \begin{flushleft}
        \begin{tabular}{ll}
            \textbf{Input:} & A distribution $\empiricaldist$ that is $\epsilon_S$-close to some MHR-like distribution $\valuedist$ in terms of $d_q$;\\
            & the optimal price $p^*$ of the underlying MHR-like value distribution. \\[1ex]
            \textbf{Output:} & An MHR-like distribution $\widetilde{\empiricaldist}$ that is $6\epsilon_S$-close to $\valuedist$ in terms of norm $d_q$.
        \end{tabular}
    \end{flushleft}
    %
    %
    \begin{enumerate}
        \item Construct a sequence of distributions as follows:
        \begin{enumerate}
            \item $\empiricaldist_1$: Subtract $\epsilon_S$ from the quantile of every value; round it to $0$ if it becomes negative.
                \label{step:empirical_1}
            %
            %
            \item $\empiricaldist_2$: Increase the quantile of $p^*$ to that in $\empiricaldist$ plus $\epsilon_S$.
                \label{step:empirical_2}
            \item $\empiricaldist_3$: Increase the quantile of each value until it meets the convex hull of the revenue curve.
                \label{step:empirical_3}
        \end{enumerate}
        \item Return $\widetilde{\empiricaldist} = \empiricaldist_3$.
    \end{enumerate}
\end{algorithm*}

\begin{figure*}[t]
    \begin{subfigure}{.24\textwidth}
        \includegraphics[width=\textwidth]{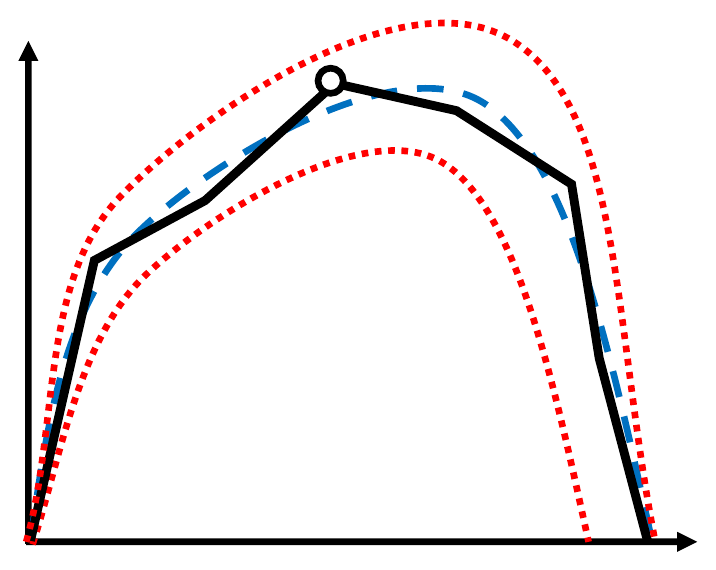}
        \caption{$\empiricaldist$ (empirical)}
        \label{fig:empiricaldist}
    \end{subfigure}
    \begin{subfigure}{.24\textwidth}
        \includegraphics[width=\textwidth]{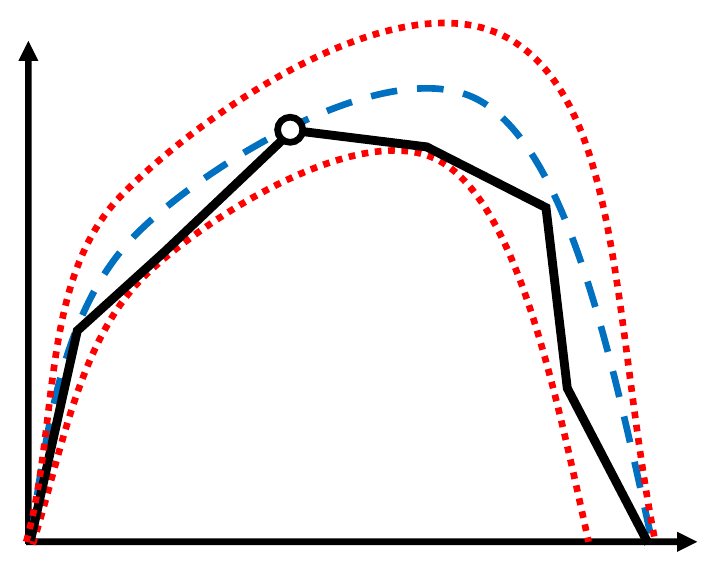}
        \caption{$\empiricaldist_1$ (dominated)}
        \label{fig:empiricaldist_1}
    \end{subfigure}
    %
    %
    \begin{subfigure}{.24\textwidth}
        \includegraphics[width=\textwidth]{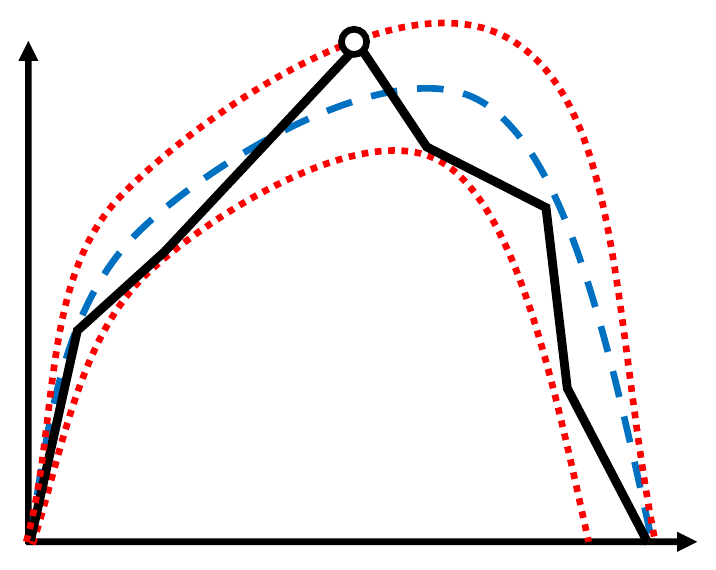}
        \caption{$\empiricaldist_2$ (strengthening $p^*$)}
        \label{fig:empiricaldist_2}
    \end{subfigure}
    \begin{subfigure}{.24\textwidth}
        \includegraphics[width=\textwidth]{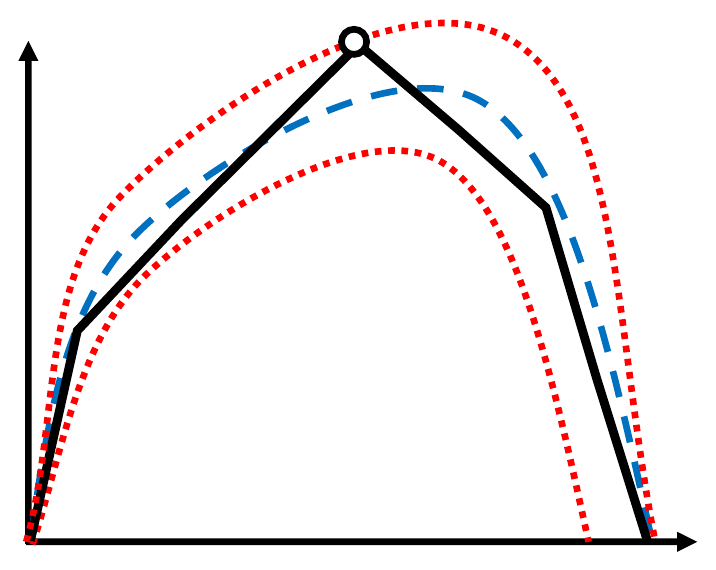}
        \caption{$\widetilde{\empiricaldist} = \empiricaldist_3$ (ironing)}
        \label{fig:empiricaldist_3}
    \end{subfigure}
    \caption{Illustrative figures (in terms of the revenue curves in the quantile space) of how Algorithm~\ref{alg:find-nearby-mhr} finds a nearby MHR-like distribution. The dashed curve is the revenue curve of the true distribution. The dotted curves are those obtain by subtracting/adding $O(\epsilon_S)$ to the quantiles of every value (rounded to $0$ or $1$ if necessary). The bold black curve corresponds to the distributions maintained by the algorithm.}
    \label{fig:find-nearby-mhr}
\end{figure*}
This subsection explains how we can find an MHR-like distribution that is $6\epsilon_S$-close to the true distribution in polynomial time, given an empirical distribution that is $\epsilon_S$-close.
Our algorithm will guess which price is the optimal one w.r.t.\ the true distribution by brute-force.
Given each guess $p^*$, we try to find a nearby MHR-like distribution conditioned on $p^*$ being optimal using Algorithm~\ref{alg:find-nearby-mhr}.
It constructs a sequence of distributions with the last one being the desired output, provided that the guess of $p^*$ is correct.

The first distribution $\empiricaldist_1$ is obtained by subtracting $\epsilon_S$ from the quantiles of all values, rounding up to $0$ if necessary.
The quantile of each value $v$ w.r.t.\ the empirical distribution $\empiricaldist$ is within a $[-\epsilon_S, \epsilon_S]$ window near that w.r.t.\ the true distribution;
in contrast, the quantile w.r.t.\ $\empiricaldist_1$ is within a $[-2\epsilon_S, 0]$ window.
Distribution $\empiricaldist_1$ is dominated by the true distribution in the sense of first-order stochastic dominance and, hence, is called the dominated empirical distribution (e.g., \cite{GHZ19, RS16}).


Then, we construct the second distribution $\empiricaldist_2$ by increasing the quantile of the conjectured monopoly price $p^*$ to ensure that it is at least as large as in the true distribution.
The purpose of this step is to ensure the strong concavity property near the monopoly price.

Finally, we run an ironing step to restore concavity. 
The result, i.e., $\empiricaldist_3$, is the final output.
See Figure~\ref{fig:find-nearby-mhr} for an illustrative picture of relations between the revenue curves of the sequence of distributions constructed by the algorithm.

\begin{lemma}
    \label{lem:robustify-mhr-like-projection}
    Given any distribution $\empiricaldist$ that is $\epsilon_S$-close to an MHR-like distribution $\valuedist$, and the monopoly price $p*$ of $\valuedist$, Algorithm~\ref{alg:find-nearby-mhr} computes in polynomial time an MHR-like distribution $\widetilde{\empiricaldist}$ that is $6\epsilon_S$-close to $\valuedist$ in terms of norm $d_q(\cdot, \cdot)$.
\end{lemma}

\begin{proof}
    The running time is trivial, noting that convex hulls can be computed in polynomial time.
    It remains to show that the output is an MHR-like distribution, and is $6\epsilon_S$-close to $\empiricaldist$.

    \paragraph{MHR-like -- concavity:}
    This part follows by the definition of the algorithm~\ref{alg:find-nearby-mhr} (step~\ref{step:empirical_3}).
     
    \paragraph{MHR-like -- strong concavity near monopoly price:}
    We will show that (1) $p^*$ is the monopoly price of the final distribution $\widetilde{\empiricaldist}$, and that (2) for any price $p$ whose quantile w.r.t.\ $\widetilde{\empiricaldist}$ is $q$, we have:
    \[
        \big( 1 - \tfrac{1}{4}(q^* - q)^2 \big) \cdot \Rev \big( p^*, \widetilde{\empiricaldist} \big) - \Rev \big( p , \widetilde{\empiricaldist} \big) \ge 0
        ~,
    \]
    \[
        \big( 1 - \tfrac{1}{4}(q^* - q)^2 \big) \cdot p^* q^* - p q \ge 0
        ~.
    \]

    We first prove them for distribution $\empiricaldist_2$.
    Note that the quantiles of all values other than $p^*$ are at most their counterparts in the true distribution $\valuedist$ (step~\ref{step:empirical_1}), while the quantile of $p^*$ is at least that in $\valuedist$ (step~\ref{step:empirical_2}).
    Hence, $p^*$ is also the monopoly price for $\empiricaldist_2$.
    Further, the LHS of the above inequality is increasing in $q^*$ (fixing any $p$, $p^*$, and $q$), and decreasing in $q$ (fixing any $p$, $p^*$, and $q^*$).
    Since that the inequality holds for the true distribution, and that $q^*$ weakly increases and $q$ weakly decreases compared to the true distribution, it also holds for $\empiricaldist_2$.
    Finally, we argue that, moving from $\empiricaldist_2$ to $\widetilde{\empiricaldist} = \empiricaldist_3$, i.e., the ironing step, will not make $p^*$ suboptimal, as the highest point of the revenue curve will not be ironed.
    The inequality will also continue to hold, because the first term is a concave function of $q$ that stays the same while the second term becomes the convex hull of the counter part in $\valuedist_2$. 
    If a concave function dominates another one, it dominates its convex hull as well.

    \paragraph{MHR-like -- large monopoly sale probability:}
    It follows from the fact that the quantile, i.e., the sale probability, of $p^*$ weakly increases by our construction.

    \paragraph{MHR-like -- small revenue and welfare gap:}
    First, imagine that the quantile of $p^*$ is only increase to that in the true distribution $\valuedist$, while the quantiles of other values weakly decrease by our construction.
    Then, the optimal revenue stays the same and the social welfare weakly decreases and therefore the small gap property continues to hold.
    Then, we further increase the quantile of $p^*$ to that in our final output distribution $\widetilde{\empiricaldist}$.
    As a result, the social welfare as well as the optimal revenue increase by the same amount.
    This will not change the fact that their gap is small.

    \paragraph{$6\epsilon_S$-close:}
    This is true by definition up to distribution $\empiricaldist_2$.
    It remains to show that the ironing step will not make the quantile of some value $v$ exceeds its quantile w.r.t.\ the true distribution $\valuedist$ by more than $2\epsilon_S$.
    For any ironing that does not involved $p^*$, this is trivial: 
    since both endpoints of the ironed interval are below the revenue curve of the true distribution, which is concave, the entire ironed interval is below that as well.
\begin{figure*}[t]
        \centering
        \begin{subfigure}{.4\textwidth}
            \includegraphics[width=\textwidth]{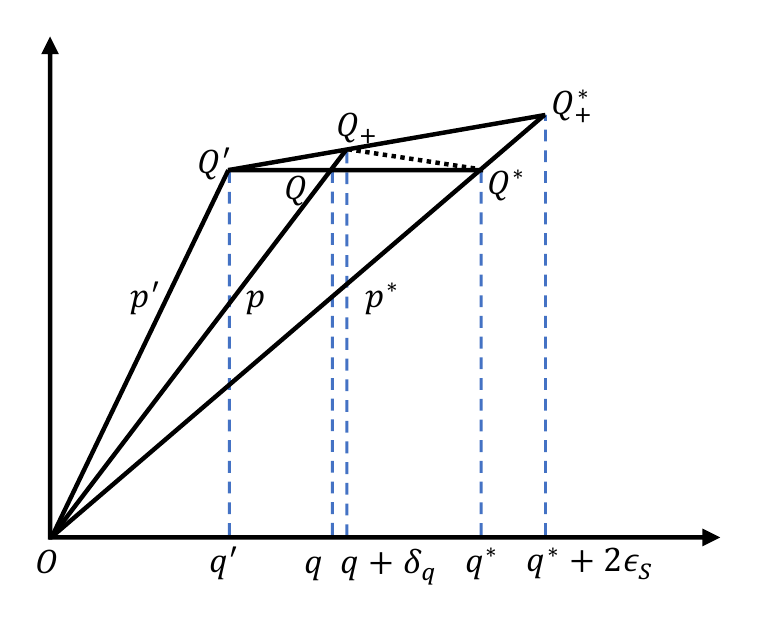}
            \caption{}
            \label{fig:find-nearby-mhr-5}
        \end{subfigure}
        \begin{subfigure}{.4\textwidth}
            \includegraphics[width=\textwidth]{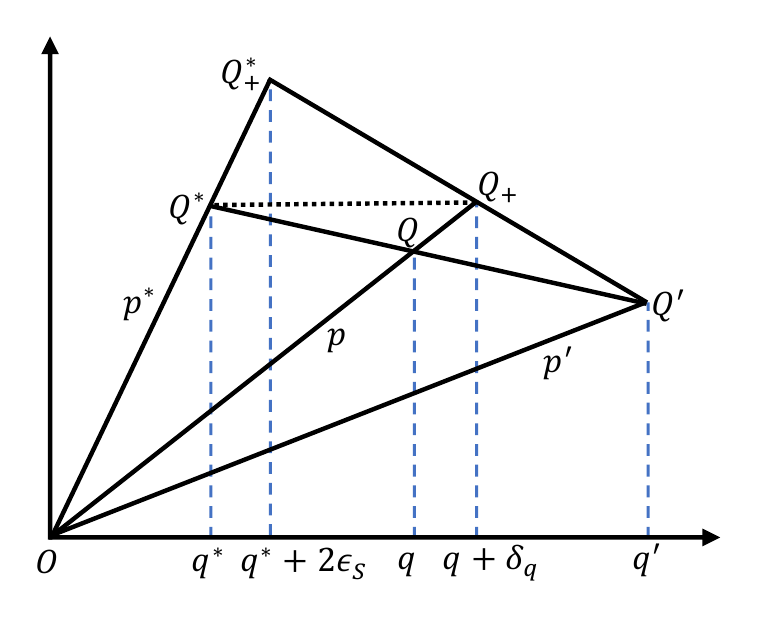}
            \caption{}
            \label{fig:find-nearby-mhr-6}
        \end{subfigure}
        \caption{Illustrative pictures for the proof of Lemma~\ref{lem:robustify-mhr-like-projection}}
    \end{figure*}
    Next, consider an ironed interval with $p^*$ being one of the endpoints. 
    Suppose $p^*$ is the right end point, i.e., there is another price $p' > p^*$.
    See Figure~\ref{fig:find-nearby-mhr-5} for an illustrative picture of the argument below.
    The quantile of $p'$ w.r.t.\ $\widetilde{\empiricaldist}$ is upper bounded by that w.r.t.\ the true distribution $\valuedist$, which we denote as $q'$;
    The quantile of $p^*$, on the other hand, is upper bounded by that w.r.t.\ the true distribution, denoted as $q^*$, plus $2\epsilon_S$ by our construction.
    Then, consider any price $p$ between $p^*$ and $p'$. 
    We can lower bound the quantile of $p$ as a function of $p'$, $p^*$, $q'$, and $q^*$, by the concavity of the revenue curve of $\valuedist$.
    We denote this quantile by $q$:
    $Q' = (q', p'q')$, $Q = (q, pq)$, and $Q^* = (q^*, p^*q^*)$ lie on the same line.
    Finally, $p$'s quantile in the final distribution, by definition, is determined by the ironing step.
    We denote this quantile by $q + \delta_q$:
    $Q' = (q', p'q')$, $Q_+ = (q+\delta_q, p (q+\delta_q))$, and $Q^*_+ = (q^*+2\epsilon_S, p^*(q^*+2\epsilon_S))$ lie on the same line.
    It suffices to bound $\delta_q$. 
    To do so, further let $O = (0, 0)$ be the origin.
    We have:
    \begin{equation}
        \label{eqn:find-nearby-mhr}
        \frac{\delta_q}{q} 
        = \frac{\mathrm{area}(Q'Q_+Q^*)}{\mathrm{area}(OQ'Q^*)} 
        \le \frac{\mathrm{area}(Q'Q^*_+Q^*)}{\mathrm{area}(OQ'Q^*)}  
        = \frac{2\epsilon_S}{q^*}
        ~.
    \end{equation}
    

    The other case, when $p^*$ is the left endpoint is almost verbatim up to the point of Eqn.~\eqref{eqn:find-nearby-mhr}.
    See Figure~\ref{fig:find-nearby-mhr-6} for an illustrative picture.
    The only catch is that we now have $q' \ge q^*$.
    Fortunately, $q'$, which is at most $1$, cannot much bigger since $q^* \ge \frac{1}{e}$ (Condition~\ref{con:monopoly-sale-prob} of Definition~\ref{def:mhr-like-distributions}).
    As a result, we get a weaker bound that $\delta_q \le e \cdot 2\epsilon_S \le 6 \epsilon_S$.
\end{proof}

\end{document}